\documentclass[12pt]{iopart}

\usepackage{mathdots}
\usepackage{xcolor}
\usepackage{iopams}
\usepackage{doi}
\usepackage{hyperref}
\usepackage{url}

\usepackage[hyperref]{ntheorem}
\usepackage{graphicx}
\usepackage{amssymb}
\usepackage{dsfont}

\theoremheaderfont{\normalfont \bfseries}
\newtheorem{theorem}{Theorem}[section]
\newtheorem{lemma}{Lemma}[section]
\newtheorem{corollary}{Corollary}
\newtheorem{prop}{Proposition}
\newtheorem{definition}{Definition}[section]
\theorembodyfont{\upshape}
\newtheorem*{proof}{Proof}

\begin{document}
	\title{On symmetric Tetranacci polynomials in mathematics and physics}
	\author{Nico G. Leumer\,
	} 
	\address{$^1$Donostia International Physics Center (DIPC), Paseo Manuel de Lardizabal 4, E-20018 San Sebastián, Spain}
	\address{$^2$ Université de Strasbourg, CNRS, Institut de Physique et Chimie des Matériaux de Strasbourg, UMR 7504, F-67000 Strasbourg, France}
	\address{$^3$ Institute for Theoretical Physics, University of Regensburg, 93053 Regensburg, Germany}
	
	\ead{nico.leumer@dipc.org}
	\begin{abstract}
	In this manuscript, we introduce (symmetric) Tetranacci polynomials $\xi_j$ as a twofold generalization of ordinary Tetranacci numbers, considering both non unity coefficients and generic initial values. We derive a complete closed form expression for any $\xi_j$ with the key feature of a decomposition in terms of generalized Fibonacci polynomials. For suitable conditions, $\xi_j$ can be understood as the superposition of standing waves. The issue of Tetranacci polynomials originated from their application in condensed matter physics. We explicitly demonstrate the approach for the spectrum, eigenvectors, Green's functions and transmission probability for an atomic tight binding chain exhibiting both nearest and next nearest neighbor processes. We demonstrate that in topological trivial models, complex wavevectors can form bulk states as a result of the open boundary conditions. We describe how effective next nearest neighbor bonding is engineered in state of the art theory/ experiment exploiting onsite degrees of freedom and close range hopping. We argue about experimental tune ability and on-demand complex wavevectors. 
	\end{abstract}

	\submitto{\jpa}
	\maketitle
	\section{Introduction}
	Undoubtedly one of the most famous sequences are the Fibonacci numbers $f_n$, defined recursively by $f_{n+1} = f_n + f_{n-1}$ for $n\ge 1$ and $f_0 = 0, f_1 = 1$ \cite{Vajda2008, Hoggatt1969, Horadam1961, Anđelić2020}. As noticed by Horadam in the midst of the last century, generalizations require either altered initial values or alternatively a modified recursion formula \cite{Horadam1961}. For instance, Webb and Parberry did the former and studied Fibonacci polynomials $F_n$ obeying ($n\ge 1$) $F_{n+1} = x\,F_n + F_{n-1}$, $F_0 = 0,\,F_1 = 1$ \cite{Webb1969}. Half a decade later, Hoggatt Jr. and Long defined generalized Fibonacci polynomials $\mathcal{F}_n$ \cite{Hoggatt1974} as ($n\ge 1$)
	\begin{eqnarray}\label{equation: generalized Fibonacci polynomials}
		\mathcal{F}_{n+1} = x\,\mathcal{F}_n +y\, \mathcal{F}_{n-1},\quad \mathcal{F}_0 = 0,\,\mathcal{F}_1=1,
	\end{eqnarray}
	while {\"O}zvatan and Pashaev substituted $\mathcal{F}_{0,1}$ by generic initial values $G_{0,1}$ \cite{Ozvatan2017}.
	
	Extending the recursion range in \eref{equation: generalized Fibonacci polynomials} from two to three yields Tribonacci numbers or Tribonacci polynomials depending on the coefficients and supposing properly chosen initial values \cite{Feinberg1963, Jishe2011, Bicknell1973, Hoggatt1973, Waddill1967}. Subsequently, the first notion of Tetranacci numbers, where the next element of the sequence is formed by the previous four, appeared (to our best knowledge) in \cite{Feinberg1963}. Since then, Tetranacci or Tetranacci-like sequences were studied in many variations up to modern days, cf. \cite{Tasci2017, Soykan2020, Waddill1992, Singh2014, Hathiwala2019} in order to mention only a few. The most generic form of what we call hereinafter a \emph{Tetranacci polynomials} $t_n$ is ($n\ge 0$)
	\begin{eqnarray}\label{equation: Tetranacci numbers}
		t_{n+2}  = x_1\,t_{n+1}+ x_0\,t_{n} + x_{-1}\,t_{n-1} + x_{-2}\,t_{n-2}
	\end{eqnarray}
	with some initial values $t_{-2},\ldots,\,t_1$ and given coefficients $x_1,\ldots,\,x_{-2}$ was previously presented in \cite{Soykan2020}. 
	
	In contrast to the generic case, we focus on symmetric Tetranacci polynomials (cf. \eref{equation: Tetranacci recursion formula} below) recovered from \eref{equation: Tetranacci numbers} for $x_{-2} = -1$ and $x_1 = x_{-1}$ but still generic $x_{0,1}$. This particular choice of coefficients seems arbitrary; it is not. Rather, the eigenvectors of ($a,\,b,\,c,\, \alpha,\, \beta\in \mathbb{R}$)
	\begin{eqnarray}\label{equation: Toeplitz like matrix}
		M(\alpha, \beta) = \left[\matrix{%
			c-\alpha & b & a\cr	
			b & c  & b & a\cr	
			a & b & c & b & a\cr	
			&	\ddots &	\ddots &	\ddots &	\ddots &	\ddots \cr
			&&a & b & c & b & a\cr	
			&	&&a & b & c & b \cr	
			&&	&&a & b & c-\beta &
		}ß\hspace{-0.3cm}\right]_{N\times N}
	\end{eqnarray}
	symmetric Toeplitz matrices $	M(\alpha =0, \beta=0)$ are composed of symmetric Tetranacci polynomials as we shall discuss. Finite $\alpha$ or $\beta$ merely modifies the boundary constraints of eigenvectors from the pure Toeplitz case \cite{Leumer2020}.
	
	Besides mathematical interests, symmetric Tetranacci appear also in condensed matter theory. The reason for this is, that we physicists consider most often particular hermitian systems to which we refer as being "translation invariant". Thus, the model's physics is captured by (banded) Toeplitz matrices. For $\alpha = \beta = 0$, $M$ mimics an atomic chain with nearest and next nearest neighbor hopping, while specific $\alpha\neq 0$, $\beta \neq 0$ are connected to the $X-Y$ chain in transverse magnetic field \cite{Loginov1997} or the Kitaev chain \cite{Leumer2020, Loginov1997, Kitaev-2001}. It has been shown further, that symmetric Tetranacci polynomials are linked to characteristic polynomial and the Green's functions used in quantum transport \cite{Leumer2021}. 
		
	The issue of eigenvalues of banded Toeplitz was investigated formerly in more generality \cite{Trench1985}. Particularly, the spectrum of tridiagonal Toeplitz attracted some interest \cite{Kouachi2008, Kouachi2006, Yueh2005} as their the eigenvector elements are Chebyshev \cite{daFonseca2005, daFonseca2019, Gover1994} or Fibonacci polynomials \cite{Shin1997}. Although a matrix as in \eref{equation: Toeplitz like matrix} can be generated by a product of two tridiagonal Toeplitz matrices \cite{Leumer2020}, results for the spectrum of $M$ are more complicated w.r.t. to the tridiagonal case as we shall see.
	
	One of our main contributions is to present a simple and closed form expression for generic symmetric Tetranacci polynomials, which originates from a decomposition into generalized Fibonacci polynomials (cf. \sref{section: Fibonacci decomposition}). As motivation, the fundamental conviction of physicists is that eigenstates (eigenvectors) of finite systems are given in terms of standing waves. Their form is sinusoidal and the perhaps most evident example are oscillations of a guitar string, being fixed at both ends. By the knowledge that Binet/ Binet-like forms of Fibonacci polynomials can be reshaped into a sine function \cite{Webb1969, Hoggatt1974}, the stage was set. 
	
	The manuscript is organized as follows. In \sref{section: introduction basic tetranacci}, we formally define symmetric Tetranacci polynomials and present the basic strategy to find their closed form expression. Subsequently, we introduce so called \emph{basic} Tetranacci polynomials and discuss a few of their properties. In \sref{section: Fibonacci decomposition}, we demonstrate that specific generalized Fibonacci polynomials obey also the Tetranacci recursion formula. We verify that any generic Tetranacci polynomial can be expressed in terms of those specific solutions. Then, we focus on applications in quantum physics. In \sref{section: lattice model}, we discuss the spectrum and the eigenvectors of an atomic chain owning nearest and next nearest neighbor hopping granting \eref{equation: Toeplitz like matrix} for $\alpha = \beta = 0$. Since some feature of this model require the next nearest neighbor process to dominate the close-range coupling, \sref{section: Engineering effective next nearest neighbor coupling} discusses how arbitrarily large ratios of effective couplings can be engineered experimentally. We conclude in \sref{section: conclusion}. \ref{appendix: degenerated roots} (\ref{appendix: Formulae of T-1,.., T_1 for degenerate roots}) discusses (states closed formulae for) degenerate roots of the Tetranacci recursion. Finally, \ref{appendix: quatum transport} presents the application to quantum transport. Readers mainly interested in the application may not hesitate to proceed from \sref{section: lattice model} on.
	\section{Generic properties of symmetric Tetranacci polynomials}\label{section: introduction basic tetranacci}
	\begin{definition}\label{def.: Tetranacci polynomial}
		The symmetric Tetranacci polynomial $\xi_j$ is recursively defined by 
		\begin{eqnarray}\label{equation: Tetranacci recursion formula}
			\xi_{j+2}\,=\,\zeta\,\xi_j-\xi_{j-2}\,+\,\eta\left(\xi_{j+1}+\xi_{j-1}\right), \quad j \in \mathbb{Z}
		\end{eqnarray}
		in terms of its initial values $\xi_i= g_i (\zeta,\,\eta) \in \mathbb{C}$ ($i=-2,\,\ldots,\,1$) and complex coefficients $\zeta,\,\eta$.
	\end{definition}
	Although the initial values may or may not depend themselves on $\zeta$ and/ or $\eta$, we utilize always the shorthand notation of $g_{-2},\ldots,\,g_{1}$ and $\xi_j$ respectively. For the purpose of illustration, the first few terms of the sequence read
	\numparts
	\begin{eqnarray}\label{equation: xi2}
		\xi_2&=-g_{-2}+\eta\,g_{-1}+\zeta\, g_0+\eta\,g_1,\\
		\xi_3\,&=\,-\eta  \,g_{-2}+\left(\eta^2-1\right)\,g_{-1}+\eta\left(\zeta+1\right)\, g_0+\left(\eta^2+\zeta\right)\,g_1,\\		
		\xi_4&=-\left(\eta^2+\zeta\right)\, g_{-2}+ \eta\left(\eta^2+\zeta-1\right)\,g_{-1}+\left(\zeta + 1\right)\left(\zeta-1 +\eta^2\right)\, g_0\nonumber\\ \label{equation: xi4}
		&\qquad+ \eta\left(\eta^2 +2\zeta+1\right)\,g_1
	\end{eqnarray}
	\endnumparts
	and further ones follow from \eref{equation: Tetranacci recursion formula}. Alternatively, we may also rely on the generating function $E(t)$. 
	\begin{prop} The generating function $E(t) = \sum\limits_{k=0}^\infty \xi_k\, t^k$ of symmetric Tetranacci polynomials reads
		\begin{eqnarray}\label{equation: generating function}
			E(t) = \frac{g_1\,t\,+\,g_0\,(1-\eta\,t)\,+\,g_{-1}\,(\eta \,t^2-t^3)-g_{-2}\,t^2}{1-\eta\,t-\zeta\,t^2-\eta\,t^3+t^4}.
		\end{eqnarray}
	\end{prop}
	\begin{proof} Using the definition of the generating function grants
		\begin{eqnarray}
		\fl	E(t) &= g_0\,+\,g_1\,t+t^2\,\sum\limits_{n=0}^\infty \xi_{n+2}\, t^n\nonumber\\\fl
			&= g_0\,+\,g_1\,t+t^2\,\sum\limits_{n=0}^\infty \left[\zeta\,\xi_{n}-\xi_{n-2}+\eta\left(\xi_{n+1}+ \xi_{n-1}\right)\right]\, t^n\nonumber\\\fl
			&= g_0\,+\,g_1\,t+g_{-1}t^2\,(\eta-t)-g_{-2}\,t^2-g_0\,\eta t   +E(t)\left(\zeta t^2-t^4+\eta t^3+\eta\,t\right).
		\end{eqnarray}
		Here, we substituted \eref{equation: Tetranacci recursion formula} in the last term and all sums were completed properly in order to provide $E(t)$. \hfill $\square$
	\end{proof}
	Perhaps contrary to their appearance in \eref{equation: xi2}- \eref{equation: xi4}, the looked for closed form of $\xi_j$ is rather simple. Since the intention of our mindset is the applicability, we aim at a particular expression for $\xi_j$, namely \eref{equation: final closed form for tetranaccis} below, demanding the introduction of specific solutions to \eref{equation: Tetranacci recursion formula},  hereinafter referred to as the basic Tetranacci polynomials. 
	\begin{definition}\label{definition:  the basic tetranacci polynomials} The basic Tetranacci polynomials $\mathcal{T}_i(j)$ ($i=-2,\ldots,\,1$) satisfy \eref{equation: Tetranacci recursion formula} for generic $j\in \mathbb{Z}$. Their initial values are summarized by
		\begin{eqnarray}\label{equation: selective property of the basic tetranacci polynomials}
			\mathcal{T}_i(j)\,=\,\delta_{ij},\quad i,j =-2,\ldots,\,1.
		\end{eqnarray}
		and $\delta_{ij}$ denotes the Kronecker-Delta\footnote{The Kronecker-Delta is defined as $\delta_{nl}=1$ for $n=l$ and zero otherwise.}. We call \eref{equation: selective property of the basic tetranacci polynomials} the selective property of $\mathcal{T}_i(j)$ as becomes evident hereinafter.
	\end{definition}
	For the purpose of illustration, \Tref{table: first few values of the Basic Tetranacci polynomials} presents the first few terms of $\mathcal{T}_i(j)$. Additional ones can be anticipated from \eref{equation: xi2}-\eref{equation: xi4}. The primary advantage of the basic Tetranacci polynomial resides in the fact that the arbitrary initial values of $\xi_j$ and the recursion formula \eref{equation: Tetranacci recursion formula} separate by means of the selective property \eref{equation: selective property of the basic tetranacci polynomials}. A similar strategy was pursued in \cite{Singh2014} for Tetranacci numbers. Nevertheless, we then have \emph{initially} to deal with four symmetric Tetranacci polynomials rather than only one.
	\begin{corollary}\label{corollary: closed form for xij}
		Any symmetric Tetranacci polynomial $\xi_j$ can be written as ($j\in \mathbb{Z}$)
		\begin{eqnarray}\label{equation: final closed form for tetranaccis}
			\xi_j\,=\,\sum\limits_{i=-2}^1\,g_i\,\mathcal{T}_i(j)
		\end{eqnarray}
		for generic $\eta$, $\zeta\in \mathbb{C}$ and complex initial values $\xi_i=g_i$, $i=-2,\ldots,\,1$.
	\end{corollary}
	\begin{proof} Due to the linearity of the Tetranacci recursion formula, any linear combination of solutions also satisfies \eref{equation: Tetranacci recursion formula}; thus, the l.h.s. of \eref{equation: final closed form for tetranaccis} is a symmetric Tetranacci polynomial. Hence, in case that \eref{equation: final closed form for tetranaccis} holds already for the initial values, this relation is true for generic integer $j$. Indeed, we find that $(j=-2,\,\ldots,\,1)$
		\begin{eqnarray}
			\xi_j\,=\,\sum\limits_{i=-2}^1\,g_i\,\mathcal{T}_i(j)\,=\,	\,\sum\limits_{i=-2}^1\,g_i\,\delta_{ij}\,=\,g_j
		\end{eqnarray}
		is correct, substituting \eref{equation: selective property of the basic tetranacci polynomials} in the intermediate step. \hfill $\square$
	\end{proof}
	\begin{table}[t]
		\centering\includegraphics[width = 0.5 \textwidth]{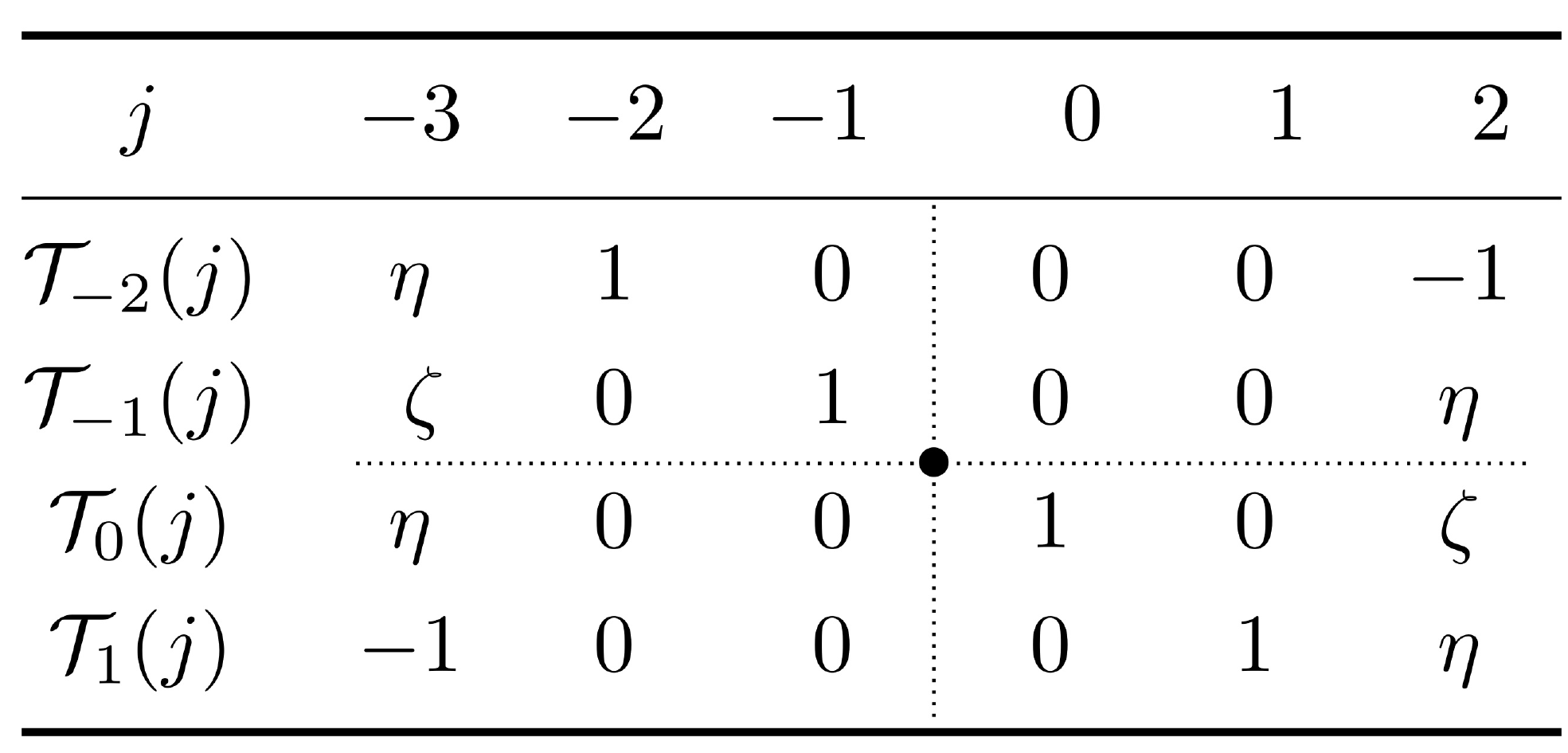}
	\caption{Basic Tetranacci polynomials $\mathcal{T}_i(j)$ for $j=-3,\,\ldots,\,2$. The central columns ($j=-2,\,\ldots,\,1$) provide the intitial values according to \eref{equation: selective property of the basic tetranacci polynomials}. The inversion point ($\bullet$) proposes the relations $\mathcal{T}_{0}(j)\,=\,\mathcal{T}_{-1}(-1-j)$ and $\mathcal{T}_{1}(j)\,=\,\mathcal{T}_{-2}(-1-j)$ for arbitrary $j$, $\eta$, $\zeta$ proven later in Lemma \ref{lemma: inversion relation of Tij}.}
	\label{table: first few values of the Basic Tetranacci polynomials}
	\end{table}
	Naturally, the description of $\xi_j$ in terms of  $\mathcal{T}_i(j)$ is not specific for Tetranacci polynomials and small modifications in both Definition \ref{definition:  the basic tetranacci polynomials} and \eref{equation: final closed form for tetranaccis} extend this strategy to arbitrary (linear) recursive problems.
	
	The basic Tetranacci polynomials inherit some specific traits originating from their particular initial values, as can be anticipated from \Tref{table: first few values of the Basic Tetranacci polynomials}. More importantly though, we find interconnections between $\mathcal{T}_{1}$ ($\mathcal{T}_{-1}$) and $\mathcal{T}_{-2}$ ($\mathcal{T}_{0}$) reducing effectively the number of involved quantities. Actually, the Lemmata \ref{lemma: inversion relation of Tij}, \ref{lemma: further interconnection of the basic Tetranacci polynomials} below even demonstrate that $\mathcal{T}_{-1}$, $\mathcal{T}_{0}$ and $\mathcal{T}_{1}$ can be constructed solely from $\mathcal{T}_{-2}$.
	\begin{lemma}\label{lemma: inversion relation of Tij}
		The basic Tetranacci polynomials $\mathcal{T}_{i}(j)$ ($i=-2,\,\ldots,\,1$) obey 
		\numparts
		\begin{eqnarray}
			\label{equation: inversion relation 1}
			\mathcal{T}_{1}(j)\,=\,\mathcal{T}_{-2}(-1-j),\\
			\label{equation: inversion relation 2}
			\mathcal{T}_{0}(j)\,=\,\mathcal{T}_{-1}(-1-j),\\
			\label{equation: inversion relation 3}
			\mathcal{T}_{-2}(j)\,=\,\mathcal{T}_{1}(-1-j),\\
			\label{equation: inversion relation 4}
			\mathcal{T}_{-1}(j)\,=\,\mathcal{T}_{0}(-1-j),
		\end{eqnarray}
		\endnumparts
		for all $j\in \mathbb{Z}$ and generic $\zeta,\,\eta \in \mathbb{C}$.
	\end{lemma}
	\begin{proof}
		Notice that once the validity of \eref{equation: inversion relation 1} (\eref{equation: inversion relation 2}) is shown, \eref{equation: inversion relation 3} (\eref{equation: inversion relation 4}) follows automatically by setting $l= -1-j$ and renaming $l\rightarrow j$ afterwards. Since the proofs of \eref{equation: inversion relation 1}, \eref{equation: inversion relation 2} are similar, we focus only on the former. The presented values in \Tref{table: first few values of the Basic Tetranacci polynomials} imply the validity of \eref{equation: inversion relation 1} already for $j=-2,\,-1,\,0$: $ \mathcal{T}_1(j) = \mathcal{T}_{-2}(-1-j) = 0 $. At $j=-3$, we find $\mathcal{T}_1(j) = \mathcal{T}_{-2}(-1-j) = -1$. Assuming that \eref{equation: inversion relation 1} holds already for some integers $n-2,\,n-1,\,n,\,n+1$, we are left to demonstrate \eref{equation: inversion relation 1} at $n+2$ ($n-3$) for increasing (decreasing) indices. Since $\mathcal{T}_{1}(j)$, $\mathcal{T}_{-2}(j)$ are symmetric Tetranacci polynomials, \eref{equation: Tetranacci recursion formula} gives
		\begin{eqnarray}\label{equation: inversion prop1}
			\mathcal{T}_1(n+2)\,=\,\zeta\,\mathcal{T}_1(n)\,-\,\mathcal{T}_1(n-2)\,+\,\eta\,\left[\mathcal{T}_1(n+1)\,+\,\mathcal{T}_1(n-1)\right]	
		\end{eqnarray}
		at $j=n$. Similarly at $j=-1-n$, we find
		\begin{eqnarray}\label{equation: inversion prop2}
		\quad~\fl \mathcal{T}_{-2}(-3-n)\,=\,\zeta\,\mathcal{T}_{-2}(-1-n)\,-\,\mathcal{T}_{-2}(1-n)\,+\,\eta\,\left[\mathcal{T}_{-2}(-2-n) +  \mathcal{T}_{-2}(-n) \right],
		\end{eqnarray}
		after reordering the terms. Due to our assumption, we find that \eref{equation: inversion prop1}, \eref{equation: inversion prop2} are identical which is equivalent to $\mathcal{T}_1(j) =\mathcal{T}_{-2}(-1-j)$ at $j=n+2$. The demonstration for decreasing indices, i.e. for $n-3$, is carried out analogously by exchanging the $j+2$ and $j-2$ terms in \eref{equation: Tetranacci recursion formula}. \hfill $\square$
	\end{proof}
	\begin{lemma}\label{lemma: further interconnection of the basic Tetranacci polynomials} The basic Tetranacci polynomials
		$\mathcal{T}_{-2}(j),\,\ldots,\,\mathcal{T}_{1}(j)$ obey
		\numparts
		\begin{eqnarray}
			\label{equation: T-2 is odd in j}
			\mathcal{T}_{-2}(j)\,&=\,-\mathcal{T}_{-2}(-j),\\
			\label{equation: T-1j in terms of T-2}
			\mathcal{T}_{-1}(j)\,&=\,\mathcal{T}_{-2}(j-1)-\eta\,\mathcal{T}_{-2}(j),\\
			\label{equation: T0j in terms of T-2}
			\mathcal{T}_{0}(j)\,&=\,\eta\,\mathcal{T}_{-2}(j+1)-\mathcal{T}_{-2}(j+2),\\
			\label{equation: T1j in terms of T-2}
			\mathcal{T}_{1}(j)\,&=\,-\mathcal{T}_{-2}(j+1),
		\end{eqnarray}
		\endnumparts
		for all $j\in \mathbb{Z}$ and generic $\zeta,\,\eta \in \mathbb{C}$.
	\end{lemma}
	\begin{proof}
		We focus first on \eref{equation: T-2 is odd in j}, where a proof for $j\ge 0$ is sufficient. Apparently, \eref{equation: T-2 is odd in j} is already valid for $j=0,\,1,\,2$ (cf. \Tref{table: first few values of the Basic Tetranacci polynomials}). For $j=3$, we obtain $\mathcal{T}_{-2}(3) = -\eta$ from \eref{equation: Tetranacci recursion formula}, i.e. we find $\mathcal{T}_{-2}(3) = -\eta = -\mathcal{T}_{-2}(-3)$. Assuming \eref{equation: T-2 is odd in j} is true for $n-2,\,n-1,\,n,\,n+1$ ($n\ge 2$) we demonstrate its validity at $n+2$. \eref{equation: Tetranacci recursion formula}, implies
		\begin{eqnarray}
		\qquad\fl\mathcal{T}_{-2}(n+2) &= \zeta\, \mathcal{T}_{-2}(n) - \mathcal{T}_{-2}(n-2) +\eta\left[\mathcal{T}_{-2}(n+1)+ \mathcal{T}_{-2}(n-1)\right],\\
		\qquad\fl\mathcal{T}_{-2}(-n-2) &= \zeta\, \mathcal{T}_{-2}(-n) - \mathcal{T}_{-2}(2-n) +\eta\left[\mathcal{T}_{-2}(-1-n)+ \mathcal{T}_{-2}(1-n)\right].
		\end{eqnarray}
		Due to our assumption, the two expressions differ only by a sign. Thus, \eref{equation: T-2 is odd in j} is valid. Next, we focus on \eref{equation: T-1j in terms of T-2}. $\mathcal{T}_{-2}(j-1)$ is apparently a solution to \eref{equation: Tetranacci recursion formula}, since $\mathcal{T}_{-2}(j)$ is a symmetric Tetranacci polynomial, i.e.  $\mathcal{T}_{-2}(j-1)-\eta\,\mathcal{T}_{-2}(j)$ is one as well due to the linearity of \eref{equation: Tetranacci recursion formula}. Hence, the latter has only to satisfy the selective property of $\mathcal{T}_{-1}(j)$ for \eref{equation: T-1j in terms of T-2} to be correct. Indeed, we have (cf. \Tref{table: first few values of the Basic Tetranacci polynomials})
		\numparts
		\begin{eqnarray}
			j=-2:&\quad \mathcal{T}_{-2}(-3)-\eta\,\mathcal{T}_{-2}(-2) = \eta -\eta = 0  \equiv \mathcal{T}_{-1}(-2),\\
			j=-1:&\quad \mathcal{T}_{-2}(-2)-\eta\,\mathcal{T}_{-2}(-1) = 1 - 0 = 1  \equiv \mathcal{T}_{-1}(-1),\\
			j=0:&\quad \mathcal{T}_{-2}(-1)-\eta\,\mathcal{T}_{-2}(0) = 0 - 0 =  0 \equiv \mathcal{T}_{-1}(0),\\
			j=1:&\quad \mathcal{T}_{-2}(0)-\eta\,\mathcal{T}_{-2}(1) = 0 - 0 =  0 \equiv \mathcal{T}_{-1}(1).
		\end{eqnarray}
		\endnumparts
		The correctness of \eref{equation: T0j in terms of T-2} is a direct consequence of \eref{equation: inversion relation 2}, \eref{equation: T-2 is odd in j}, \eref{equation: T-1j in terms of T-2}:
		\begin{eqnarray}
			\mathcal{T}_0(j)= \mathcal{T}_{-1}(-1-j)  &=\mathcal{T}_{-2}(-2-j)-\eta\,\mathcal{T}_{-2}(-1-j)\nonumber\\&=\eta\,\mathcal{T}_{-2}(j+1)- \mathcal{T}_{-2}(2+j),
		\end{eqnarray}
		while \eref{equation: T1j in terms of T-2} follows from \eref{equation: inversion relation 1}, \eref{equation: T-2 is odd in j}:
		\begin{eqnarray}
			\mathcal{T}_1(j) = \mathcal{T}_{-2}(-1-j) =- \mathcal{T}_{-2}(j+1).
		\end{eqnarray}
		Thus, the statements are correct. \hfill $\square$
	\end{proof}
	In the view of Corollary \ref{corollary: closed form for xij} and the Lemmata \ref{lemma: inversion relation of Tij}, \ref{lemma: further interconnection of the basic Tetranacci polynomials}, the closed form expression of an arbitrary Tetranacci polynomial $\xi_j$ demands merely the one of $\mathcal{T}_{-2}$. Yet the final result for $\mathcal{T}_{-2}$ presented in Theorem \ref{theorem: Fibonacci decomposition of T-2} below requires some preparation.
	
	Furthermore, \eref{equation: final closed form for tetranaccis} is also interesting when studying algebraic properties of $\xi_j$. For instance, we find ($j \in \mathbb{Z}$)
	\begin{eqnarray}
		\xi_{-1-j} = g_{-2} \,\mathcal{T}_1(j) +g_{-1} \,\mathcal{T}_0(j) + g_{0} \,\mathcal{T}_{-1}(j) + g_{1} \,\mathcal{T}_{-2}(j)
	\end{eqnarray}
	by imposing \eref{equation: inversion relation 1}-\eref{equation: inversion relation 4} on \eref{equation: final closed form for tetranaccis}. Thus, $\xi_{-1-j} = \xi_j$ holds in case that $g_{-2} = g_1$ and $g_{-1} = g_0$ are true. Similar properties of $\xi_j$ may follow, once they have been proven for $\mathcal{T}_{-2}(j),\,\ldots,\,\mathcal{T}_1(j)$.
	
	Although \eref{equation: tetranacci plane wave form} below is rather trivial from the mathematical point of view, i.e. that $\xi_j$ can be written as linear combination of complex entities, the Lemma summarizes (to our best knowledge) all relations necessary to diagonalize symmetric Toeplitz matrices of bandwidth two. As long as \eref{equation: tetranacci plane wave form} is valid, this relation is consistent with the in solid state physics famous Bloch's theorem without touching further details \cite{Flensberg2004, Leumerthesis}. Nevertheless, the consequences imposed by the quantities $S_{1,2}$ (defined in Lemma \ref{lemma: exponential form for tetranaccis} below) are of fundamental importance for us.
	\begin{lemma}\label{lemma: exponential form for tetranaccis}
		Any symmetric Tetranacci polynomial can be expressed as
		\begin{eqnarray}\label{equation: tetranacci plane wave form}
			\xi_j\,=\,A\,e^{\rmi \theta_1\,j}\,+\,B\,e^{-\rmi \theta_1\,j}\,+\,C\,e^{\rmi \theta_2\,j}\,+\,D\,e^{-\rmi \theta_2\,j},
		\end{eqnarray}
		provided that $S_1 \neq S_2$ and $S_{1,2}\neq \pm 2$ hold true where $S_{1,2} = (\eta\pm \sqrt{\eta^2+4(\zeta+2)}\,)/2$. In \eref{equation: tetranacci plane wave form}, we introduced $\theta_{1,2}\in \mathbb{C}$ defined by $2\cos(\theta_{1,2})= S_{1,2}$. The coefficients $A$, $B$, $C$, $D$ are set implicitly by $\xi_i = g_{i}$, $i=-2,\ldots,\,1$.
	\end{lemma}
	\begin{proof} The announced result is found straightforwardly by the power law ansatz $\xi_j\propto r^j$ ($r\neq 0$) on \eref{equation: Tetranacci recursion formula}. After substituting the ansatz and dividing by $r^j\neq 0$, we arrive at the characteristic equation: 
		\begin{eqnarray}\label{equation: characteristic equation for r}
			r^2 \,+\, \frac{1}{r^2}\,-\,\zeta\,-\,\eta\left(r+\frac{1}{r}\right)\,=\,0.
		\end{eqnarray}
		Its peculiar form suggests to introduce the variable $S = r+r^{-1}$, granting in turn the quadratic equation $S^2 - \eta S-\zeta-2 = 0$, whose zeros are
		\begin{eqnarray}\label{equation: definition S12}
			S_{1,2} = \frac{\eta\pm \sqrt{\eta^2+4(\zeta+2)}}{2}.
		\end{eqnarray}
		Solving $S = r+r^{-1}$ for $r$ at $S=S_{1,2}$ yields
		\begin{eqnarray}\label{equation: definition rpml in the power ansatz}
			r_{\pm l}\,=\,\frac{S_l\pm \sqrt{S_l^2-4}}{2},\quad l = 1,\,2,
		\end{eqnarray}
		possessing the properties of $r_{+ l} r_{- l} = 1$ and $r_{+ l} +r_{- l} = S_l$ for $l=1,2$. In case of $S_1 \neq S_2$,  $S_{1,2}\neq \pm 2$, we can thus express $\xi_j$ as their linear combination:
		\begin{eqnarray}\label{equation: bloch like theorem for non-degenerate roots}
			\xi_j=A\, r_{+1}^j+B\,r_{-1}^j +C \,r_{+2}^j+D\, r_{-2}^j.
		\end{eqnarray}
		The coefficients $A$, $B$, $C$, $D$ are to be set by the initial values $g_{-2},\ldots, g_1$ of $\xi_j$. Introducing $\theta_{1,2}$ by $2\cos(\theta_{1,2})= S_{1,2}$ grants $r_{\pm l} = \exp(\pm \rmi \theta_{l})$. \hfill $\square$
	\end{proof}
	Any degeneracy of the roots $r_{\pm 1,2}$ alters \eref{equation: tetranacci plane wave form} qualitatively, i.e. the closed form expression for $\xi_j$ will change, and we refer here to \ref{appendix: degenerated roots} for further details. Perhaps contrary to the impression of the reader that we apply next Lemma \ref{lemma: exponential form for tetranaccis} to determine $\mathcal{T}_{-2}$ or $\xi_{-2}$ resulting in a Binet-like form, similar to the one for Tetranacci numbers in \cite{Soykan2020}, we follow a different strategy. In fact specific generalized Fibonacci polynomials also satisfy \eref{equation: Tetranacci recursion formula} out of which $\mathcal{T}_{-2}$ can be constructed. One has still to distinguish the cases of: i) $S_1\neq S_2$, ii) $S_1=S_2$, but $S_1^2\neq 4$ and iii) $S_1=S_2$, $S_1^2= 4$.
	\section{The Fibonacci decomposition}\label{section: Fibonacci decomposition}
	Generalized Fibonacci polynomials, defined here according to \eref{equation: generalized Fibonacci polynomials} (cf. \cite{Hoggatt1974}), are closely related to symmetric Tetranacci polynomials. The trivial limit is $\eta =0$, where \eref{equation: Tetranacci recursion formula} simplifies to \makebox{$\xi_{j+2} = \zeta \xi_j - \xi_{j-2}$}; thus, separating even and odd indices $j$. Defining then $v_l = \xi_{2l}$ ($u_l = \xi_{2l+1}$) immediately grants $u_{l+1} = \zeta u_l -u_{l-1}$, $v_{l+1} = \zeta v_l -v_{l-1}$. Yet, even for $\eta\neq 0$ we find that specific symmetric Tetranacci polynomials obey simultaneously a two term recursion formula.
	\begin{theorem} \label{theorem: hidden Fibonacci polynomials}
		The generalized Fibonacci polynomial $\varphi_{l}(j)$ ($l=1,\,2$), set by
		\begin{eqnarray}\label{equation: hidden Fibonacci polynomials}
			\varphi_{l}(j+1) \,=\,S_l\,\varphi_{l}(j)-\,\varphi_{l}(j-1),\quad j \in \mathbb{Z}
		\end{eqnarray}
		with $S_{1,2}\,=\,(\eta\pm \sqrt{\eta^2\,+\,4\left(\zeta +2\right)})/2$ and initial values $\varphi_{l}(0) =0$, $\varphi_{l}(1)=1$ is a symmetric Tetranacci polynomial.
	\end{theorem}
	\begin{proof}
		For the sake of clarity, we suppress the index $l=1,2$ in the following. The proposed statement follows straightforwardly by assuming initially that $\varphi_{l}(j)$ obeys $\varphi (j+1) \,=\,x\,\varphi (j)+\,y\,\varphi(j-1)$ ($j \in \mathbb{Z}$) for arbitrary (complex) $x$, $y$ and initial values $f_0$, $f_1$. In order to be a symmetric Tetranacci polynomial, $\varphi_{l}(j)$ has to satisfy also \eref{equation: Tetranacci recursion formula}. Replacing in \eref{equation: Tetranacci recursion formula} all terms carrying the indices $j+2$, $j+1$ grant
		\begin{eqnarray}\label{equation: intermediate Fibonacci recursion}
			(x^2 +y - \eta\,x -\zeta)\,\varphi (j) \,=\,(\eta y +\eta -xy)\, \varphi (j-1)\, -\, \varphi(j-2).
		\end{eqnarray}
		Comparing the coefficients between \eref{equation: intermediate Fibonacci recursion} and our ansatz for $\varphi$ sets $y=-1$ immediately. In turn, $\eta y +\eta -xy = x$ holds without restrictions on $x$. Instead, the latter is set by $1 = x^2 +y - \eta\,x -\zeta \equiv  x^2 - \eta x - \zeta - 1$ after substituting $y$. The associated quadratic equation has the two roots $S_{1,2} = (\eta\pm \sqrt{\eta^2\,+\,4\left(\zeta +2\right)}\,) /2$ introduced earlier in Lemma \ref{lemma: exponential form for tetranaccis}. Thus, $\varphi (j)$ obeys \eref{equation: Tetranacci recursion formula} for the announced coefficients. As $\varphi (j)$ is well defined by $f_0$, $f_1$ and \eref{equation: hidden Fibonacci polynomials}, the initial values $\varphi (-2)$, $\varphi (-1)$, $\varphi (0) =f_0$, $\varphi (1)=f_1$ for the Tetranacci recursion in \eref{equation: Tetranacci recursion formula} are fixed. Hence, $\varphi (j)$ satisfies the definition of symmetric Tetranacci polynomials. Without loss of generality, we choose $f_0 = 0$ and $f_1 = 1$ for simplicity.\hfill $\square$
	\end{proof}
	Notice that Theorem \ref{theorem: hidden Fibonacci polynomials} is an implication: An arbitrary symmetric Tetranacci polynomial $\xi_j$ will not obey \eref{equation: hidden Fibonacci polynomials} due to its generic initial values $g_{-2},\ldots, g_{1}$. For instance, we may choose $g_{-2} = \varphi_1(-2) + \epsilon$, $g_{-1} = \varphi_1(-1)$ and $g_0 = \varphi_1(0)=0$, $g_1 = \varphi_1(1)=1$ for $\epsilon>0$. 
	
	Further, \Tref{table: some values of varphi12} exposes the first few values of $\varphi_{1,2}$, from where we deduce the next proposition before we turn to their closed form expression.
	\begin{table}[tb]
	\begin{indented}\item[]
		\begin{tabular}{@{}cccccccc}     \br	 
			$j$ 			& 	$-3$	&	$-2$	&	$-1$ 	&	$0$ 	&	$1$ 	&	$2$ &	$3$	\\\hline 
			$\varphi_l(j)$ \rule{0pt}{3ex} 	&$-(S_l^2-1)$	&	$-S_l$	&	$-1$ 	&	$0$ & $1$ 	&	$S_l$ 	&	$S_l^2-1$	\\
			\br
		\end{tabular}
	\end{indented}
	\caption{\label{table: some values of varphi12}The first terms of $\varphi_{l}(j)$ ($l=1,\,2$) for $j=-3,\,\ldots,3$.}	
	\end{table}
	\begin{prop}\label{proposition: varphi is odd in j} The polynomials $\varphi_l(j)$ ($l=1,2$) satisfy $\varphi_l(j)=-\varphi_l(-j)$ for all $j\in \mathbb{Z}$.
	\end{prop}
	\begin{proof} \eref{equation: hidden Fibonacci polynomials} and the initial values $\varphi_l(0)=0$, $\varphi_l(1)=1$ yield $\varphi_l(-1)=-1$. Thus, the statement is correct for $j=0,\,1$. Assuming that $\varphi_l(j)=-\varphi_l(-j)$ holds already at $n,\,n+1$ ($n\in \mathbb{N}_0$), \eref{equation: hidden Fibonacci polynomials} states that
		\begin{eqnarray}
			\varphi_l (n+2) \,=\,S_l\,\varphi_l(n+1)-\,\varphi_l (n)\,=\,-\left[S_l \varphi_l(-n-1)-\varphi_l(-n)\right]
		\end{eqnarray}
		is true. Exchanging $\varphi_l(j+1)$, $\varphi_l(j-1)$ in \eref{equation: hidden Fibonacci polynomials} and using that $n\ge 0$ implies $-n-1<-n$, identifies $S_l \varphi_l(-n-1)-\varphi_l(-n) =\varphi_l(-n-2)$.\hfill $\square$
	\end{proof}
	\begin{prop}\label{proposition: explicit closed form for varphi12}
		The explicit closed form expression for $\varphi_l(j)$ ($l=1,2$, $j \in \mathbb{Z}$) 
		\begin{eqnarray}\label{equation: Binet form for varphi}
			\varphi_{l}(j) = \frac{r_{+l}^j-r_{-l}^j}{r_{+l}-r_{-l}}
		\end{eqnarray}
		is Binet-like whenever $S_l\neq \pm 2$. In terms of $\theta_{1,2}$, one has $\varphi_{l} = \sin(\theta_l j)/ \sin(\theta_l)$ \cite{Webb1969, Hoggatt1974}. For $S_l=\pm 2$, we find
		\begin{eqnarray}\label{equation: linear form for varphi}
			\varphi_{l}(j) = j \left(\frac{S_l}{2}\right)^{j+1}.
		\end{eqnarray}
	\end{prop}
	\begin{proof}
		First, we focus on $S_l\neq \pm 2$. Although \eref{equation: Binet form for varphi} and its version in terms of $\theta_l$ is known as the closed form expression for generalized Fibonacci polynomials (cf. \cite{Webb1969, Hoggatt1974}), we re-derive it for completeness and in order to better demonstrate the changes imposed by $S_l^2 = 4$ afterwards. Using the ansatz $\varphi_l \propto r^j$ ($r\neq 0$) on \eref{equation: hidden Fibonacci polynomials}, we find $r^2-S_l\,r+1 =0$ after dividing by $r^{j-1}$. The two roots are $r_{\pm l} = (S_l\pm \sqrt{S_l^2-4}) /2$ from \eref{equation: definition rpml in the power ansatz} in Lemma \ref{lemma: exponential form for tetranaccis} and $S_l\neq \pm 2$ implies $\varphi_l(j) = \alpha \,r_{+l}^j +\beta\, r_{-l}^j$. The coefficients $\alpha$, $\beta$ are set by $\varphi_l(0)=0$ and $\varphi_l(0)=1$ as $\alpha = r_{+l} -r_{-l}$, $\beta = -\alpha$ granting \eref{equation: Binet form for varphi}. Introducing $\theta_l$ as $2\cos(\theta_l)=S_l$ turns \eref{equation: Binet form for varphi} into $\sin(\theta_l j)/\sin(\theta_l)$.
		
		For $S_l= \pm 2$, the two roots $r_{\pm l}$ become degenerate: $r_{+l} = r_{-l} = S_l/2$. Hence, the linear combination $\varphi_l(j) = \alpha \,r_{+l}^j +\beta\, r_{-l}^j = \tilde{\alpha}\, r_{+l}^j$ becomes insufficient to properly account for the two initial values of $\varphi_l(j)$. Since the recursion formula \eref{equation: hidden Fibonacci polynomials} does not change qualitatively at $S_l^2 =4$, one misses in fact the second solution $j\, r_+^j$ as demonstrated in the following. Substituting the ansatz $\varphi_l(j)\propto j\, r_{+l}^j$ into \eref{equation: hidden Fibonacci polynomials} and reordering according to powers in $j$ grants
		\begin{eqnarray}
			j\,(r_{+l}^2-S_l\,r_{+l}+1) + r_{+l}^2-1 = 0
		\end{eqnarray}
		where the bracket vanishes since $r_{+l}$ is the root of $r^2-S_l\,r +1$. The second term vanishes due to $r_{+l} =S_l/2 = \pm 1$ found from $S_l^2 =4$. Thus, $\varphi_l(j) = \tilde{\alpha}\, r_{+l}^j +  \tilde{\beta} j\, r_{+l}^j = (\tilde{\alpha} + \tilde{\beta}\,j)r_{+l}^j$ is true. The initial values of $\varphi_l(j)$ set $\tilde{\alpha} = 0$,  $\tilde{\beta} = 1/r_+ = r_+$ and \eref{equation: linear form for varphi} is found. \hfill $\square$
	\end{proof}
	Similar as for "ordinary" symmetric Tetranacci polynomials $\xi_j$, the situation of $S_1 = S_2$ offers further special solutions to \eref{equation: Tetranacci recursion formula} apart from only $\varphi_{1,2}(j)$ (cf. Lemma \ref{lemma: exponential form for tetranaccis} and \ref{appendix: degenerated roots}). The following lemma is the last intermediate step, before we finally turn to one of the main results of the article: The decomposition of $\mathcal{T}_{-2}(j)$ (and thus any symmetric Tetranacci polynomial) in terms of the generalized Fibonacci polynomials $\varphi_{1,2}(j)$.
	\begin{lemma}\label{lemma: jvarphi symmetric Tetranacci for S1  =S2} For $S_1 = S_2$ (and $S_1 = \pm 2$) also $j \varphi_{1,2}(j)$ ($j^2 \varphi_{1,2}(j)$) is a symmetric Tetranacci polynomial.
	\end{lemma}
	\begin{proof} The situation of $S_1 = S_2$ implies $\varphi_1(j) = \varphi_2(j)$ for all $j \in \mathbb{Z}$ (cf. Theorem \ref{theorem: hidden Fibonacci polynomials}) and thus; we demonstrate the statement for only $j \varphi_1(j)$ and $j^2 \varphi_1(j)$. After substituting $j \varphi_1(j)$ into \eref{equation: Tetranacci recursion formula}, 
		and reordering, we arrive at
		\begin{eqnarray}
			0 &= j\left\{\varphi_1(j+2)-\zeta \varphi_1(j)+\varphi_1(j-2)-\eta\left[\varphi_1(j+1)+\varphi_1(j-1)\right]\right\}\nonumber\\
			&\qquad + 2\left[\varphi_1(j+2) - \varphi_1(j-2)\right]-\eta \left[\varphi_1(j+1) - \varphi_1(j-1)\right]\nonumber\\
			\label{equation: jvarphi is solution}
			& =2\left[\varphi_1(j+2) - \varphi_1(j-2)\right]-\eta \left[\varphi_1(j+1) - \varphi_1(j-1)\right],
		\end{eqnarray}
		where the curly bracket is identically zero, since $\varphi_1(j)$ is a symmetric Tetranacci polynomial (cf. Theorem \ref{theorem: hidden Fibonacci polynomials}). So far, we have not imposed $S_1 = S_2$. According to \eref{equation: definition S12}, we find $S_1 = \eta/2$ and using \eref{equation: hidden Fibonacci polynomials} twice shows that \eref{equation: jvarphi is solution} is indeed satisfied
		\begin{eqnarray}
			\fl\qquad\qquad 2\left[\varphi_1(j+2) - \varphi_1(j-2)\right]-\eta \left[\varphi_1(j+1) - \varphi_1(j-1)\right]\nonumber\\
			\fl\qquad\qquad = 2 \underbrace{\left[\varphi_1(j+2)\,- S_1\,\varphi_1(j+1)\right] }_{= -\varphi_1(j)}
			+ 2\underbrace{\left[S_1\, \varphi_1(j-1)\,-\varphi_1(j-2)\right]}_{=\varphi_1(j)} =0,
		\end{eqnarray}
		i.e. $j\varphi_1(j)$ is a symmetric Tetranacci polynomial. Notice that the constraint $S_1 = S_2$ is essential, as otherwise only $\eta = S_1+S_2$ is correct. Then, we may write $S_1 =S_2 + \epsilon $ ($\epsilon \neq 0$) and \eref{equation: jvarphi is solution} is not satisfied.
		
		Next, in order for $j^2\varphi_1(j)$ to obey \eref{equation: jvarphi is solution} the additional constraint $S_1^2 = 4$ is mandatory. Substituting $j^2\varphi_1(j)$ into \eref{equation: Tetranacci recursion formula} and reordering the terms afterwards grants
		\begin{eqnarray}
			0\,&=\,j^2\left\{\varphi_1(j+2)-\zeta\varphi_1(j)+\varphi(j-2)-\eta \left[\varphi_1(j-1)+\varphi_1(j+1)\right]\right\}\nonumber\\
			&\qquad + 2j\left\{2\left[\varphi_1(j+2) - \varphi_1(j-2)\right]-\eta \left[\varphi_1(j+1) - \varphi_1(j-1)\right]\right\}  	\nonumber\\
			&\qquad + 4\,\varphi_1(j+2) + 4\, \varphi_1(j-2)-\eta\,\left[\varphi_1(j+1)+\varphi_1(j-1)\right]\nonumber\\
			&= 4\,\varphi_1(j+2) + 4\, \varphi_1(j-2)-\eta\,\left[\varphi_1(j+1)+\varphi_1(j-1)\right],
		\end{eqnarray}
		where the first two lines drop since $\varphi_1(j)$, $j\varphi_1(j)$ are symmetric Tetranacci polynomials. Due to \eref{equation: hidden Fibonacci polynomials} and $S_1 =\eta/2$, we find
		\begin{eqnarray}
			&4\,\varphi_1(j+2) + 4\, \varphi_1(j-2)-\eta\,\left[\varphi_1(j+1)+\varphi_1(j-1)\right]\nonumber\\ 
			&= 4\left[S_1\,\varphi_1(j+1)-\varphi_1(j)\right]
			+ 4\left[S_1\,\varphi_1(j-1)-\varphi_1(j)\right]
			-\eta\,S_1\varphi_1(j)\nonumber\\
			& = 2 \left[2 S_1^2-S_1\frac{\eta}{2}-4 \right]\varphi_1(j)\nonumber\\
			& = 2 \left[S_1^2-4\right]\varphi_1(j),
		\end{eqnarray}
		being zero only for $S_1^2=4$ at generic $j\in \mathbb{Z}$.  \hfill $\square$
	\end{proof}
	Next, we construct $\mathcal{T}_{-2}$ in terms of $\varphi_{1,2}(j)$ and $j\varphi_{1,2}(j)$, $j^2\varphi_{1,2}(j)$ when the proper conditions are met.
	\begin{theorem}\label{theorem: Fibonacci decomposition of T-2}
		The closed form expression of $\mathcal{T}_{-2}(j)$ is 	
		\vspace{0.1cm}
		\begin{eqnarray}\label{equation: closed form expression of T-2}
			\mathcal{T}_{-2}(j)  = \left\{
			\matrix{
				\frac{\varphi_2(j)- \varphi_1(j)}{S_1-S_2}, & S_1\neq S_2\cr
				\frac{(1-j)\,\varphi_1(j+1)\,+\, (1+j)\,\varphi_1(j-1)}{S_1^2-4}, & S_1 = S_2, S_1^2\neq 4\cr
				\frac{S_1\,(1-j^2)}{12} \,\varphi_1(j), & S_1 = S_2, S_1^2 =  4}
			\right. .
		\end{eqnarray}
	\end{theorem}\vspace{0.1cm}
	\begin{proof}
		For $S_1\neq S_2$, we have that $\varphi_{1,2}(j)$ satisfy \eref{equation: Tetranacci recursion formula}, while $j\varphi_{1,2}(j)$, $j^2\varphi_{1,2}(j)$ do not (cf. Theorem \ref{theorem: hidden Fibonacci polynomials} and Lemma \ref{lemma: jvarphi symmetric Tetranacci for S1  =S2}). Since $S_1\neq S_2$ also implies $\varphi_1(j)\neq \varphi_2(j)$ (cf. \eref{equation: hidden Fibonacci polynomials}) for all $j \in \mathbb{Z}\setminus\{0,\,1\}$ and due to the linearity of \eref{equation: Tetranacci recursion formula},  $[\varphi_2(j)- \varphi_1(j)]/(S_1-S_2)$ is a non-trivial solution to \eref{equation: Tetranacci recursion formula}. Hence, the statement is correct provided that $[\varphi_2(j)- \varphi_1(j)]/(S_1-S_2)$ inherits the selective property of $\mathcal{T}_{-2}(j)$. Using \Tref{table: some values of varphi12}, we have indeed that
		\numparts
		\begin{eqnarray}
		\fl \qquad	j=-2:&  \quad \frac{\varphi_2(-2)- \varphi_1(-2)}{S_1-S_2}\,=\,\frac{-S_2-(-S_1)}{S_1-S_2} = 1\equiv \mathcal{T}_{-2}(-2),\\
		\fl \qquad	j=-1:&  \quad \frac{\varphi_2(-1)- \varphi_1(-1)}{S_1-S_2}\,=\,\frac{-1-(-1)}{S_1-S_2} = 0\equiv \mathcal{T}_{-2}(-1),\\
		\fl\qquad	j=0:&  \quad \frac{\varphi_2(0)- \varphi_1(0)}{S_1-S_2}\,=\,\frac{0-0}{S_1-S_2} = 0\equiv \mathcal{T}_{-2}(0),\\
		\fl\qquad	j=1:&  \quad \frac{\varphi_2(1)- \varphi_1(1)}{S_1-S_2}\,=\,\frac{1-1}{S_1-S_2} = 0\equiv \mathcal{T}_{-2}(1),
		\end{eqnarray}
		\endnumparts
		holds true. Turning to the case of $S_1=S_2$ but $S_1^2\neq 4$,
		we first find $\varphi_1(j) = \varphi_2(j)$ (cf. \eref{equation: hidden Fibonacci polynomials}) for all $j\in \mathbb{Z}$, but also $j\varphi_1(j)$ satisfies now \eref{equation: Tetranacci recursion formula} due to Lemma \ref{lemma: jvarphi symmetric Tetranacci for S1  =S2}. Since $\varphi_1(j\pm1 )$, $(1\pm j)\varphi_1(j\pm1 )$ are apparent solutions of \eref{equation: Tetranacci recursion formula}, we find that
		\begin{eqnarray}
		\fl \qquad	(1-j)\,\varphi_1(j+1)\,+\, (1+j)\,\varphi_1(j-1) &= 
			2\,\varphi_1(j+1)- (j+1)\, \varphi_1(j+1)\nonumber \\
			&\quad +2\,\varphi_1(j-1)\,+ (j-1)\, \varphi_1(j-1)
		\end{eqnarray}
		is one too. In addition, we have (cf. \Tref{table: some values of varphi12})
		\numparts
		\begin{eqnarray}
		\fl \qquad	j=-2:& \quad \frac{3\,\varphi_1(-1)\,-\,\varphi_1(-3)}{S_1^2-4} = \frac{-3+(S_1^2-1)}{S_1^2-4} = 1 \equiv \mathcal{T}_{-2}(-2),\\ 
		\fl \qquad	j=-1:& \quad \frac{2\,\varphi_1(0)\,+\,0 \,\varphi_1(-2)}{S_1^2-4} = \frac{0+0}{S_1^2-4} = 0 \equiv \mathcal{T}_{-2}(-1),\\
		\fl \qquad	j=0: & \quad \frac{\varphi_1(1)\,+\,\varphi_1(-1)}{S_1^2-4} = \frac{1-1}{S_1^2-4} = 0 \equiv \mathcal{T}_{-2}(0),\\
		\fl \qquad	j=1: & \quad \frac{0\, \varphi_1(2)\,+\,2\,\varphi_1(0)}{S_1^2-4} = \frac{0+0}{S_1^2-4} = 0 \equiv \mathcal{T}_{-2}(1),
		\end{eqnarray} 
		\endnumparts
		i.e. $\mathcal{T}_{-2}(j)$ is properly constructed. Finally, in case of $S_1=S_2$ and $S_1^2=4$ also $j^2\varphi_1(j)$ satisfies \eref{equation: Tetranacci recursion formula} (Lemma \ref{lemma: jvarphi symmetric Tetranacci for S1  =S2}). Apparently, $(1-j^2)\varphi_1(j)$ vanishes at $j = \pm 1$ and also for $j=0$ due to $\varphi_1(0)=0$. For $j=-2$, we find
		\begin{eqnarray}
			\frac{S_1\,(1-4)}{12} \,\varphi_1(-2) = \frac{3S_1^2}{12} = 1 \equiv \mathcal{T}_{-2}(-2),
		\end{eqnarray}
		and the statement is correct.\hfill $\square$
	\end{proof}
	Since the expression for $\mathcal{T}_{-2}(j)$ is known to us, we can next construct the remaining basic Tetranacci polynomials by applying the Lemmata \ref{lemma: inversion relation of Tij}, \ref{lemma: further interconnection of the basic Tetranacci polynomials}.
	\begin{prop}\label{proposition: closed formula T-1 for S1 not S2}
		For generic integer $j$, we have
		\numparts
		\begin{eqnarray}\label{equation: formula T-1 for S1 not S2}
			\mathcal{T}_{-1}(j)\,&=\,\frac{\varphi_1(j+1)-\varphi_2(j+1) +S_2\,\varphi_{1}(j) -S_1\,\varphi_{2}(j)}{S_1-S_2},\\
			\mathcal{T}_{0}(j)\,&=\,\frac{S_1\,\varphi_2(j+1) -S_2\,\varphi_1(j+1)+\varphi_{2}(j) -\varphi_{1}(j)}{S_1-S_2},\\
			\label{equation: formula T1 for S1 not S2}
			\mathcal{T}_{1}(j)\,&=\,\frac{\varphi_1(j+1) -\varphi_2(j+1)}{S_1-S_2},
		\end{eqnarray}
		\endnumparts
		supposing here $S_1\neq S_2$. The results for $S_1 = S_2$ are presented in \ref{appendix: Formulae of T-1,.., T_1 for degenerate roots}.
	\end{prop}
	\begin{proof}
		The displayed formulae follow directly by substituting \eref{equation: closed form expression of T-2} into the relations from Lemmata \ref{lemma: inversion relation of Tij}, \ref{lemma: further interconnection of the basic Tetranacci polynomials} and exploiting the properties of $\varphi_{1,2}(j)$ drawn in Proposition \ref{proposition: varphi is odd in j} and Theorem \ref{theorem: hidden Fibonacci polynomials}. Alternatively, the \eref{equation: formula T-1 for S1 not S2} - \eref{equation: formula T1 for S1 not S2} are apparently linear combinations of solutions to the recursion formula in \eref{equation: Tetranacci recursion formula} and one is left to demonstrate the respective selective property, which we delegate as exercise to the reader.\hfill $\square$
	\end{proof}
	In the beginning of this manuscript, we promised to provide a rather simple closed form expression for $\xi_j$. On first glance of $\mathcal{T}_i(j)$ ($i=-2,\,\ldots,\,1$) in Theorem \ref{theorem: Fibonacci decomposition of T-2} and Proposition \ref{proposition: closed formula T-1 for S1 not S2}, this seems wrong. However, substituting the basic Tetranacci polynomials into \eref{corollary: closed form for xij} yields ($j\in \mathbb{Z}$)
	\begin{eqnarray}\label{equation: final closed form for xi_j for S1 not S2}
		\xi_j\,&=\,\varphi_2(j)\,\frac{g_{-2}-S_1\,g_{-1}+ g_0}{S_1-S_2}- \varphi_1(j)\,\frac{g_{-2}-S_2\,g_{-1}+ g_0}{S_1-S_2}\nonumber\\
		& \quad + \varphi_1(j+1)\,\frac{g_{-1} - S_2\, g_0+g_1}{S_1-S_2} -\varphi_2(j+1)\,\frac{g_{-1} - S_1\, g_0+g_1}{S_1-S_2}
	\end{eqnarray}
	in case of $S_1\neq S_2$. Similar expressions can be anticipated also for $S_1 = S_2$, $S_1^2 \neq 4$ and $S_1 = S_2$, $S_1^2=4$ from \ref{appendix: Formulae of T-1,.., T_1 for degenerate roots}. In view of Theorem \ref{theorem: Fibonacci decomposition of T-2} and Lemma \ref{lemma: further interconnection of the basic Tetranacci polynomials}, we demonstrated explicitly the decomposition of a generic symmetric Tetranacci polynomial $\xi_j$ in terms of the generalized Fibonacci polynomials $\varphi_{1,2}$.
	
	The substitution $\varphi_{1,2}(j) = \sin(\theta_{1,2} j)/\sin(\theta_{1,2})$ ($S_{1,2}\neq \pm 2$, cf. Proposition \ref{proposition: explicit closed form for varphi12}) shows that $\xi_j$ can be seen as combination of standing waves. Particularly, we replace $\theta_{1,2} = k_{1,2} d$ in terms of wavevectors $k_{1,2}$ for physical models with lattice constant $d$. In case a boundary condition is applied, the values of $k_{1,2}$ become quantized accordingly \cite{Leumer2020}. 
	
	However, these wavevectors are generally complex. The reason is that $r_{\pm l} = \exp(\pm\rmi \theta_{l})$ ($l = 1,2$) is the polar form of $r_{\pm l}$ since $\theta_{l} = R_l+\rmi I_l$ ($R_l,\,I_l \in \mathbb{R}$) yields $r_{\pm l} = \vert r_{\pm l} \vert \,\exp(\pm \rmi R_l)$ with $\vert r_{\pm l} \vert = \exp(\mp I_l)$. Importantly, the potential complex nature of $k_{1,2}$ is independent of the model's topological classification and do not necessarily corresponds to edge modes. We verify the statement in \sref{subsection: Tetranacci arrow} below. 
	
	In the context of a tight binding model, we discuss the potential of symmetric Tetranacci polynomials in physics. During our in-depth analysis we also show current limitations of the approach as for instance, unknown non-linear identities (if existent) of $\mathcal{T}_i(j)$. 
	\section{Lattice model with nearest and next nearest neighbor hopping}
	\label{section: lattice model}
	The simplest physical model featuring Tetranacci polynomials as defined in \eref{equation: Tetranacci recursion formula} is a tight binding chain of atoms interconnected by nearest and next nearest neighbor hopping. In terms of spinless fermionic creation/ annihilation operators $c^\dagger_j$, $c_j$, the one-dimensional lattice Hamiltonian reads ($\mu,\,t_{1,2} \in \mathbb{R}$) 
	\begin{eqnarray}\label{equation: nnn chain, lattice Hamiltonian}
		\hat{H}=-\mu \sum\limits_{j=1}^N c_j^\dagger c_j - \left[\sum\limits_{n = 1}^2t_n\sum\limits_{j=1}^{N-n}  (c_{j+n}^\dagger c_j + c_j^\dagger c_{j+n}) \right]
	\end{eqnarray}
	where $\mu$ denotes an onsite energy and ($t_2$) $t_1$ abbreviates (next) nearest neighbor coupling. Figure \ref{Fig1} provides a sketch of the model for six atoms. Contrary to physical intuition, we also allow for $\vert t_2 \vert >\vert t_1 \vert$. Later in \sref{section: Engineering effective next nearest neighbor coupling}, we discuss how effective $t_2$, $t_1$ couplings of arbitrary ratio may be engineered in experiments. 
	\subsection{General approach to the spectrum}
	On the assumptions of both open boundary conditions and finite length, direct diagonalization methods can be applied on \eref{equation: nnn chain, lattice Hamiltonian}. The fermionic field $\hat{\Psi} = (c_1,\ldots,\,c_N)$, $\hat{\Psi}^\dagger = (c_1^\dagger,\ldots,\,c_N^\dagger)^\mathrm{T}$ yields the hermitian Toeplitz matrix
	\begin{eqnarray}\label{equation: nnn chain matrix}
		\mathcal{H} = \left[\matrix{%
			-\mu & -t_1 & -t_2\cr	
		-t_1 & -\mu & -t_1 & -t_2\cr	
		-t_2 & -t_1 & -\mu & -t_1 & -t_2\cr	
	&	\ddots &	\ddots &	\ddots &	\ddots &	\ddots \cr
	&&-t_2 & -t_1 & -\mu & -t_1 & -t_2\cr	
	&	&&-t_2 & -t_1 & -\mu & -t_1 \cr	
	&&	&&-t_2 & -t_1 & -\mu &
		}ß\hspace{-0.3cm}\right]_{N\times N}
	\end{eqnarray}
	satisfying $\hat{H} = \hat{\Psi}^\dagger\mathcal{H} \hat{\Psi}$. Please notice that $\mathcal{H}$ recovers $M(\alpha = 0,\,\beta = 0)$ from \eref{equation: Toeplitz like matrix}. On first glance an eigenvector $\vec{\psi}_E = \left(\xi_1,\ldots,\,\xi_N\right)$ of \eref{equation: nnn chain matrix} to eigenvalue $E \in \mathbb{R}$ obeys ($j = 3,\ldots,\,N-2$)
	\begin{eqnarray}\label{eigenvector recursion formula}
		0=(E+\mu) \,\xi_j + t_1\,(\xi_{j+1} +\xi_{j-1}) +t_2 \,(\xi_{j+2} +\xi_{j-2}),
	\end{eqnarray}
	and four boundary conditions 
	\numparts
	\begin{eqnarray}\label{equation: nnc. bound1}
		0=&	(E+\mu) \,\xi_1 + t_1\,\xi_{2}  +t_2 \,\xi_3,\\
		0=&	(E+\mu) \,\xi_2 + t_1\,(\xi_{3}+\xi_{1}) +t_2 \,\xi_4,\\
		0=&	(E+\mu) \,\xi_{N-1} + t_1\,(\xi_{N-2} + \xi_{N})+t_2\,\xi_{N-3},\\
		\label{equation: nnc. bound4}
	0=&	(E+\mu) \,\xi_{N} + t_1\,\xi_{N-1}+t_2\,\xi_{N-2}.
	\end{eqnarray}
	\endnumparts
	Since the latter four mimic \eref{eigenvector recursion formula} apart from missing terms, we define $\xi_0,\xi_{-1},\xi_{-2}\ldots$, $\xi_N,\xi_{N+1},\xi_{N+2}\ldots$ as the continuation of $\xi_1,\ldots,\,\xi_N$ while the eigenvector $\vec{\psi}_E$ remains untouched. Thus, we establish the Tetranacci recursion formula ($j\in \mathbb{Z}$)
	\begin{eqnarray}\label{equation: Tetranacci recursion nnn chain}
		\xi_{j+1}  = \zeta\,\xi_j -\xi_{j-2} + \eta\, \left(\xi_{j+1}+\xi_{j-1}\right),\quad \zeta = -\frac{E+\mu}{t_2},\,\eta = -\frac{t_1}{t_2}
	\end{eqnarray}
	and the boundary conditions reduce to $\xi_{-1} = \xi_0 = \xi_{N+1} = \xi_{N+2} = 0$ as expected from the physical perspective. Please notice though that we aim mainly on the situation of $t_2\neq 0$, as otherwise \eref{equation: nnn chain matrix} becomes tridiagonal corresponding to Fibonacci polynomials \cite{Kouachi2006, daFonseca2019}. 
	\begin{figure}[t] \centering
		\includegraphics[width = 0.65 \textwidth]{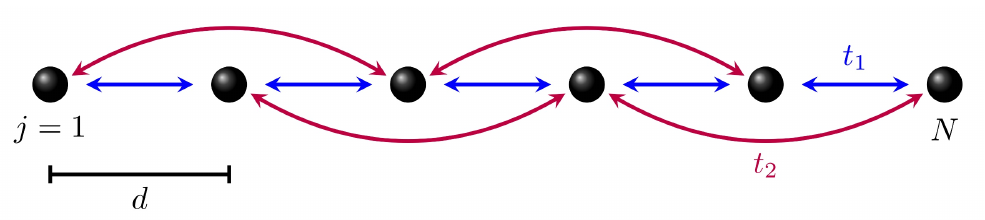}\vspace{-0.2cm}
		\caption{Atomic chain with nearest $t_1$ and next nearest neibhor hopping $t_2$.}
		\label{Fig1}
	\end{figure}

	Next, we provide the connection between \eref{equation: Tetranacci recursion nnn chain} and physical quantities such as the bulk dispersion relation and wave vectors $k_{1,2}$. Applying Lemma \ref{lemma: exponential form for tetranaccis} and renaming $\theta_{1,2} \rightarrow k_{1,2} d$ ($d$ is the inter atomic distance) sets $\xi_j = A \exp(\mathrm{i}k_1d_j) +B \exp(-\mathrm{i}k_1d_j) + C \exp(\mathrm{i}k_2d_j) +D \exp(-\mathrm{i}k_2d_j)$ in terms of two right and two left moving plane waves. Naturally, $k_{1,2}$ are related to the eigenvalue $E\equiv E(k_{1,2})$ and $E(k_1)  =E(-k_1) = E(k_2)=E(-k_2)$ has to be true. The key relations to remind are i) $S_{1,2} =r_{\pm 1,2} +r_{\pm 1,2}^{-1} $ and ii) $S^2_{1,2}-\eta S_{1,2}-\zeta -2 = 0$ relating $k_{1,2}$ and $E$ (kept inside $\zeta$). Indeed, solving for $E$ yields 
	\begin{eqnarray}\label{equation: nnn chain dispersion relation}
		E(k) = -\mu -2t_1\cos(kd)-2t_2\cos(2kd)
	\end{eqnarray}
 	the bulk dispersion relation at $k=k_{1,2}$ \cite{Leumerthesis}. Usually, \eref{equation: nnn chain dispersion relation} is found from \eref{equation: nnn chain, lattice Hamiltonian} for $N\rightarrow\infty$ and periodic boundary conditions after a Fourier analysis. We are left to demonstrate that $E(\pm k_1) = E(\pm k_2)$ is true. From \eref{equation: definition S12}, we find $S_1+S_2 = \eta$ granting $S^2_{1} -\eta S_1  = S_1^2- (S_1+S_2)S_1 = -S_1S_2 = S_2^2- (S_1+S_2)S_2 = S^2_{2} -\eta S_2$. In turn, ii) implies $\zeta = S_1^2-\eta S_1 -2 = S_2^2-\eta S_2 -2 $; thus we have $E(\pm k_1) = E(\pm k_2)$. Since the parameters $\mu$, $t_{1,2}$ and $E$ have the physical dimension of an energy, we refer to $S_1+S_2 = \eta$ as the equal energy constraint. Written explicitly, we have
	\begin{eqnarray}\label{equation: equal energy constraint}
		\cos(k_1d) + \cos(k_2d) = -\frac{t_1}{2t_2}.
	\end{eqnarray}
	By \eref{equation: nnn chain dispersion relation}, \eref{equation: equal energy constraint}, the eigenvalues $E$ demand the determination of $k_{1,2}$. Imposing the open boundary conditions on $\xi_j = A \exp(\mathrm{i}k_1d_j) +B \exp(-\mathrm{i}k_1d_j) + C \exp(\mathrm{i}k_2d_j) +D \exp(-\mathrm{i}k_2d_j)$ grants a homogeneous $4\times 4$ system of equations $M\vec{x} =0$ in $\vec{x} = (A,\,B,\,C,\,D)^\mathrm{T}$. Since $\vec{x} = \vec{0}$ yields $\vec{\psi}_E= \vec{0}$, $M$ is singular. After some algebra, we find ($2k_\pm = k_1\pm k_2$)
	\begin{eqnarray}\label{equation: quantization contraint nnn chain}
		\frac{\sin^2\left[k_+d\left(N+2\right)\right]}{\sin^2(k_+d)} = \frac{\sin^2\left[k_-d\left(N+2\right)\right]}{\sin^2(k_-d)}
	\end{eqnarray}
	setting the discrete (quantized) values of $k_{1,2}$. Notice that \eref{equation: quantization contraint nnn chain} has to be solved together with \eref{equation: equal energy constraint}. Thus, $k_{1,2}$ depends on $t_1/t_2$. Due to the matrix size, we expect $N$ pairs $(k_1,\,k_2)$ satisfying \eref{equation: quantization contraint nnn chain} and without loss of generality, their real part is restrained to the first Brillouin zone, i.e. $-\pi/d \le \mathrm{Re}(k_{1,2}) \le \pi/d$, due to $2\pi$-periodicity of \eref{equation: equal energy constraint}, \eref{equation: quantization contraint nnn chain}.
	\begin{figure}[t] \centering
		\includegraphics[width = 0.7 \textwidth]{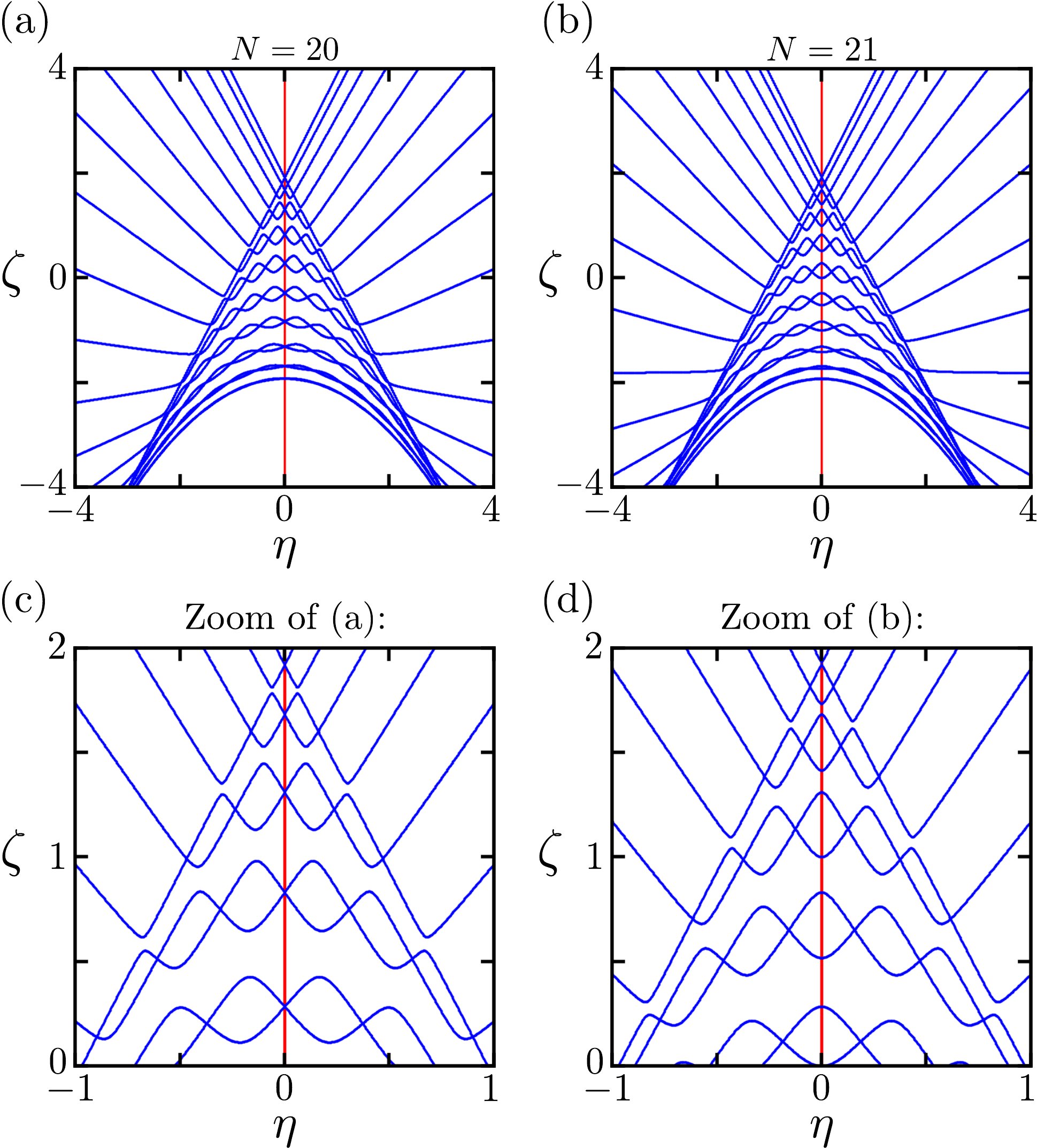}\vspace{-0.2cm}
		\caption{Numerical eigenvalues of \eref{equation: nnn chain matrix} in terms of $\zeta$, $\eta$ for $N = 20$ ($N = 21$) in (a) ((b)). (a) The spectrum features several degeneracies associated to line crossings at specific parameter values. (b) For $N$ odd, no degenerate eigenvalues appear at $\eta = 0$ (red vertical line). Generally, the spectrum is symmetric in $\eta$. (c) ((d)) Zoom of (a) ((b)).}
		\label{Fig2}
	\end{figure}

	Unfortunately though, the transcendental character of \eref{equation: quantization contraint nnn chain} denies further progress apart from limiting cases. The appearance of $N+2$ pronounces the next nearest neighbor character of the model and naive expectations as $k = n\pi/(N+1)$ ($n = 1,\ldots,\, N$) clearly fail. Even worse, $t_{1,2}$ have a strong influence. For simplicity we discuss the situation of $t_1 = 0$, when only next nearest neighbor hopping remains. Then, the Hamiltonian \eref{equation: nnn chain, lattice Hamiltonian} possess two separated sublattices consisting of only even (odd) $j$ such that \eref{equation: nnn chain matrix} becomes block diagonal in the associated basis. Each block itself is tridiagonal and mimics a nearest neighbor chain with hopping constant $t_2$. For $N$ even, each eigenvalue is twice degenerated
	\begin{eqnarray}\label{equation: nnnc t_1 = 0, N even}
		E = -\mu-2t_2 \cos \left(\frac{2n\pi}{N+2}\right), \quad n = 1,\ldots, N/2\quad 
	\end{eqnarray}
	while $N$ odd gives
	\begin{eqnarray}\label{equation: nnnc t_1 = 0, N odd1}
		E &= -\mu-2t_2 \cos \left(\frac{2n\pi}{N+3}\right), \quad n = 1,\ldots (N+1)/2\\
		\label{equation: nnnc t_1 = 0, N odd2}
		E &= -\mu-2t_2 \cos \left(\frac{2n\pi}{N+1}\right), \quad n = 1,\ldots (N-1)/2
	\end{eqnarray}
	and the general dependence in $N$ is non trivial \cite{Leumerthesis}. In \eref{equation: nnnc t_1 = 0, N even} - \eref{equation: nnnc t_1 = 0, N odd2}, we extract $k_1d$ from the cosine functions. We have $k_2 = k_1 -\pi/d$ corresponding to $k_+ = k_1 + \pi/(2d)$, $k_-=  \pi/(2d)$ and both \eref{equation: equal energy constraint} \eref{equation: quantization contraint nnn chain} are satisfied. In Figure \ref{Fig2}, the solutions from \eref{equation: nnnc t_1 = 0, N even} mark the line closings on the vertical $\eta = 0$ axis.  
	
	Next, we turn to degenerate eigenvalues. As can be seen from Figure \ref{Fig2}, they exist only for well defined ratios $t_1/t_2$. We shall derive the exact conditions (cf. \eref{equation: crossing conditions} below) and proof that only twofold degeneracies exist.
	\subsection{Degenerate eigenvalues}
	Initially, we assume that the degeneracy is $D$-fold associated to linear independent eigenvectors $\vec{\psi}_E^{(d)} = (\xi^{(d)}_1,\ldots,\xi^{(d)}_N)$ ($d = 1,\ldots,\,D\ge 2$) of \eref{equation: nnn chain matrix}. From Corollary \eref{corollary: closed form for xij}, we have $\xi^{(d)}_j = \sum_{i=-2}^1\,g_i^{(d)}\,\mathcal{T}_i(j)$ with initial values $g_i^{(d)}$. Notice that the basic Tetranacci polynomials $\mathcal{T}_i(j)$ are the same for all  $\vec{\psi}_E^{(d)}$. The boundary condition demands $\xi^{(d)}_{-1}  = \xi^{(d)}_0 = 0$ granting $\xi^{(d)}_j = g_{-2}^{(d)}\,\mathcal{T}_{-2}(j)+ g_{1}^{(d)}\,\mathcal{T}_{1}(j)$. Since only two initial values remain, the degeneracy is at best twofold as we shall see. Due to the linearity of the eigenvector equation and the Tetranacci recursion formula, linear combinations of eigenvectors correspond to one of initial values. Without loss of generality, we set $g_{-2}^{(1)} = g_{1}^{(2)} = 0$. Then, we observe that $\vec{\psi}_E^{(3)},\ldots,\,\vec{\psi}_E^{(D)}$ are composed of $\vec{\psi}_E^{(1)}$, $\vec{\psi}_E^{(2)}$, i.e. $D=2$. Next, we derive the parameter constraints.
	
	The full boundary condition was not yet imposed on $\vec{\psi}_E^{(1)}$, $\vec{\psi}_E^{(2)}$. Demanding $\xi_{N+1}^{(1,2)} = \xi_{N+2}^{(1,2)} = 0$, yields the four constraints $\mathcal{T}_{1}(N+1)=\mathcal{T}_{1}(N+2)  = \mathcal{T}_{-2}(N+1)=\mathcal{T}_{-2}(N+2) = 0$ due to $g_{1}^{(1)}\neq 0$, $g_{-2}^{(2)}\neq 0$. Exploiting \eref{equation: T1j in terms of T-2}, i.e. $\mathcal{T}_{1}(j)\,=\,-\mathcal{T}_{-2}(j+1)$ grants $\mathcal{T}_{-2}(l) = 0$ at $l = N+1,\,N+2,\,N+3$. In physics, we trust in Bloch's theorem, i.e. only the case $S_1\neq S_2$ can be relevant. We verify this assumption later by counting all crossings in order to ensure that none is missing. From \eref{equation: Binet form for varphi}, we have $\varphi_{1,2}(N+2) =\sin[k_{1,2}d(N+2)] = 0$. Seemingly, one can choose ($n_{1,2} = 1,\ldots,\,N+1$)
	\begin{eqnarray}\label{eq: crossings without selection rules}
		k_{1,2}d = \frac{n_{1,2}\pi}{N+2}
	\end{eqnarray}
	independently. However, any linear combination of $\vec{\psi}_E^{(1)}$, $\vec{\psi}_E^{(2)}$ has to satisfy the boundary conditions. Hence, \eref{equation: quantization contraint nnn chain} has to be satisfied as well, imposing selection rules on $k_{1,2}$. In terms of $n_{\rm max} = (N+2)/2$ ($n_{\rm max} = (N+1)/2$) for even (odd) $N$, we arrive at ($n = 2,\ldots,\,n_{\rm max},\, l = 1,\,\ldots,\, n-1$)
	\begin{eqnarray}\label{equation: crossing conditions}
		\left(k_+d,\,k_-d \right) &= \left(\frac{n\pi}{N+2},\,\frac{l\pi}{N+2}\right).
	\end{eqnarray} 
	The ratio for $t_1/t_2$ is set by \eref{equation: equal energy constraint} upon inserting $k_\pm$. Then, the eigenvalue follows from \eref{equation: nnn chain dispersion relation}. Notice that \eref{equation: crossing conditions} corresponds to values of $\eta\ge0$ ($t_1/t_2<0$). For $\eta\le0$, set $k_+d \rightarrow \pi-k_+d$ while $k_-d$ is kept fixed. Avoiding double counting at $\eta = 0$ for even $N$, the total number of line crossings is $N^2/4$ (($N^2-1)/4$) for $N$ even (odd) in agreement with the numerical investigation. Henceforth, the earlier restriction on $S_1\neq S_2$ was correct indeed. Next, we turn to the eigenvectors. 
	\subsection{Eigenvectors and spatial inversion symmetry}
	\label{subsection: inversion symmetry}
	Although we "determine" all eigenvectors in the following, the lack of non-linear identities (if existent) for generic $\xi_j$ or at least for $\mathcal{T}_i(j)$ is a current limitation of the approach and the solution for $\vec{\psi}_E$ appears rather unsatisfying. A similar issue may potentially arise in quantum transport depending specifically on the investigated model (cf. \ref{appendix: quatum transport}).
	
	When an eigenvalue $E$ is chosen, the eigenvectors are set by the initial values $g_{-2},\ldots,\,g_{1}$. Imposing the open boundary conditions on \eref{equation: final closed form for tetranaccis} gives $\xi_j = g_{-2}\,\mathcal{T}_{-2}(j)+ g_{1}\,\mathcal{T}_{1}(j)$ for non degenerate eigenvalues. Otherwise, we have $\xi^{(d)}_j = g_{-2}^{(d)}\,\mathcal{T}_{-2}(j)+g_{1}^{(d)}\,\mathcal{T}_{1}(j)$, $d = 1,2$ where $g_{-2}^{(d)}$, $g_{1}^{(d)}$ can be chosen independently. Therefore, we have
	\begin{eqnarray}
		\xi^{(1)}_j = g_{-2}^{(1)}\,\mathcal{T}_{-2}(j),\\
		\xi^{(2)}_j = -g_{1}^{(2)}\,\mathcal{T}_{-2}(j+1)
	\end{eqnarray}
	exploiting \eref{equation: T1j in terms of T-2}. Here, $g_{-2}^{(1)}$, $ g_{1}^{(2)}$ adopt the role of normalization constants.
	\begin{figure}[t] \centering
		\includegraphics[width = 0.6 \textwidth]{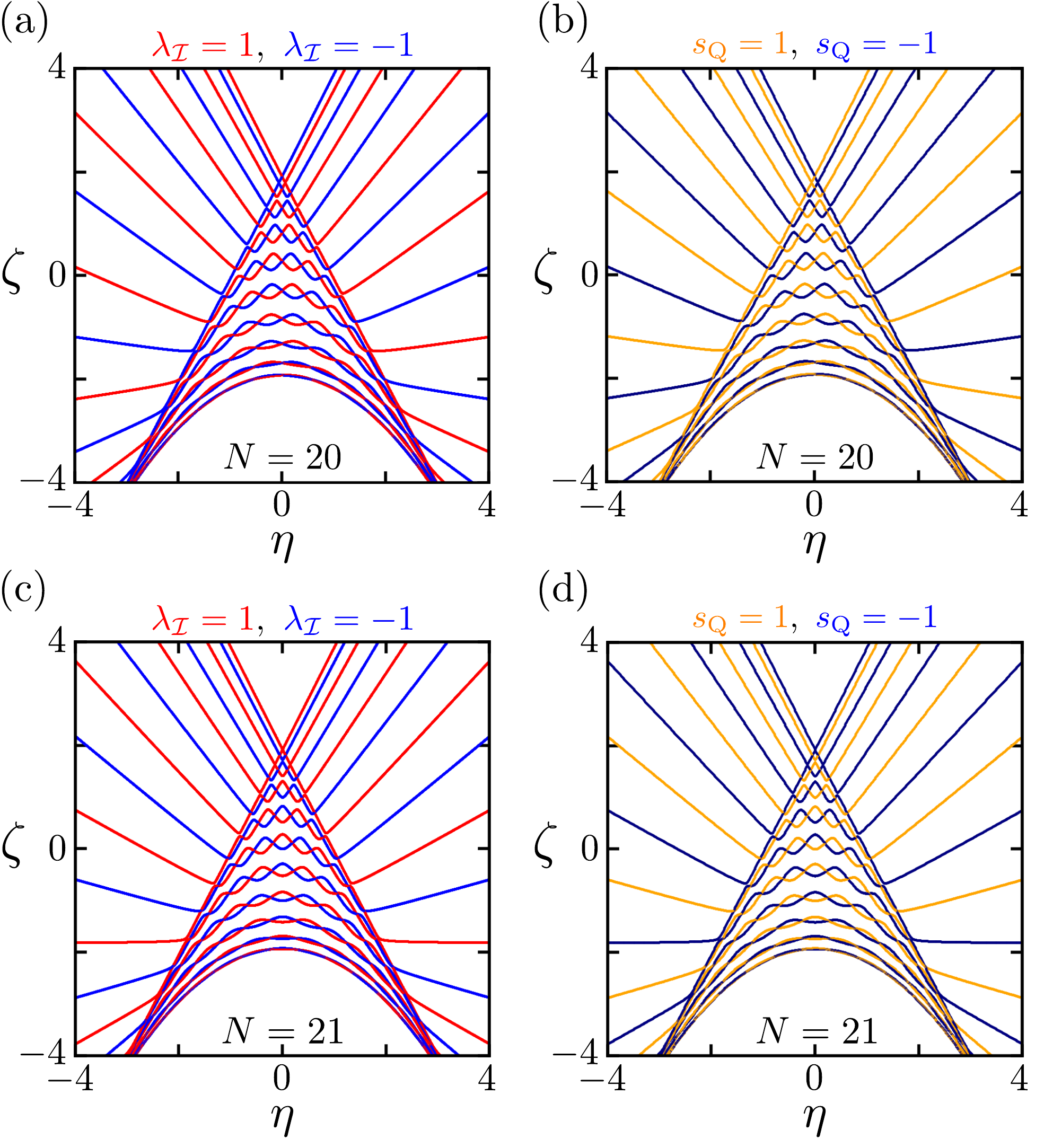}\vspace{-0.2cm}
		\caption{Relation between inversion symmetry $\lambda_{\mathcal{I}}= \pm 1$ and branch $s_{\rm  Q } = \pm 1$ in (a), (b) ((c), (d)) for even (odd) $N$. Numerically, we conclude $\lambda_{\mathcal{I}}\,s_{\rm  Q } = -1$ always, i.e. for all eigenvalues and independently of $N$.}
		\label{Fig3}
	\end{figure}
	Contrary for the non degenerate case, $\xi_{N+1} = 0$ imposes that $g_{1} = -g_{-2} \mathcal{T}_{-2}(N+1)/\mathcal{T}_{-2}(N+2)$ using \eref{equation: T1j in terms of T-2}. Both $\mathcal{T}_{-2}(N+1)$, $\mathcal{T}_{-2}(N+2)$ are finite in the non-degenerate case. Replacing $g_1$, yields
	\begin{eqnarray}\label{eq: eigenvector non degenerate case}
		\xi_j = \frac{g_{-2}}{\mathcal{T}_{-2}(N+2)}~\left[\mathcal{T}_{-2}(j)\,\mathcal{T}_{-2}(N+2)-\mathcal{T}_{-2}(N+1)\,\mathcal{T}_{-2}(j+1)\right].
	\end{eqnarray} 
	satisfying the open boundary conditions by construction. Notice though that $\xi_{N+2} = 0$ yields the quantization condition \eref{equation: quantization contraint nnn chain}. Unfortunately, further progress beyond \eref{eq: eigenvector non degenerate case} is denied due to our lack of non-linear identities for $\mathcal{T}_{-2}(j)$.
		
	The Hamiltonian \eref{equation: nnn chain, lattice Hamiltonian} possess time reversal $\tau$ and spatial inversion symmetry $\mathcal{I}$. Both imply that wavevectors appear in $\pm k$ pairs. Yet, the attempt of exploiting those fails. Explicitly, we have $\tau = \mathds{1}_N \hat{\mathcal{K}}$ ($\hat{\mathcal{K}}$ denotes the operator of complex conjugation) for spinless electrons \cite{Bernevig-2013} and
	\begin{eqnarray}
		\mathcal{I} = \left[\matrix{%
			&&&&1\cr
			&&&1\cr
			&&\iddots\cr	
			&1\cr	
			1 
		}ß\right]_{N\times N}.
	\end{eqnarray}
	Time reversal symmetry has little impact since all quantities in Eq. \eref{eq: eigenvector non degenerate case} are real and $g_{-2}$ can be chosen freely. The inversion symmetry is self-inverse $\mathcal{I}^2 = \mathds{1}_N$ with eigenvalues $\lambda_{\mathcal{I}} = \pm 1$. Since $\mathcal{I}\,\mathcal{H}\,\mathcal{I} = \mathcal{H}$ is true, the eigenvectors $\vec{\psi}_E$ fall into two categories $\mathcal{I}\,\vec{\psi}_E = \lambda_{\mathcal{I}} \vec{\psi}_E$. Even (odd) eigenvectors $\lambda_{\mathcal{I}} = +1$ ($\lambda_{\mathcal{I}} = -1$) obey $\xi_{N+1-j} = \lambda_{\mathcal{I}}\, \xi_j$. 
	
	However, this does not solve the basic problem as $\lambda_{\mathcal{I}}$ is imprinted into the quantization constraint. In fact, \eref{equation: quantization contraint nnn chain} has two branches $f(k_+) = s_{\rm Q} f(k_-)$ ($f(k) = \sin[k(N+2)]/\sin(k)$) with $s_{\rm Q} = \pm 1$. Numerically, we find $s_{\rm Q}\lambda_{\mathcal{I}} = -1$ as shown in Figure \ref{Fig3}. 
	
	The numerical investigation supports the criterion of degenerate energies. Figure \ref{Fig3} illustrates that line crossings occur between eigenvectors of opposite spatial inversion symmetry. Henceforth, the degenerate case demands $f(k_+) = f(k_-)$, $f(k_+) = - f(k_-)$ simultaneously, i.e. $f(k_+) = f(k_-) = 0$ is the only solution. Indeed, both sides of \eref{equation: quantization contraint nnn chain} vanish independently upon inserting \eref{equation: crossing conditions}.
	\subsection{The Tetranacci arrow and complex wavevectors}
	\label{subsection: Tetranacci arrow}
	The presence of only $\mathcal{I}$, $\tau$ implies that \eref{equation: nnn chain, lattice Hamiltonian} is of class AI, i.e. topological trivial according to the classification scheme of Altand and Zirnbauer \cite{Altland1997}. Thus, we neither expect the existence of complex wavevectors nor of topologically protected edgestates. However, complex values of $k_{1,2}$ are necessary in order to account properly for $t_2\rightarrow 0$. The spectrum effectively separates into two entities associated with only real wavevectors or alternatively one complex and one real wavevector. In Figure \ref{Fig4} (a),  the former is depicted in red. We refer to this region as the Tetranacci arrow.
	\begin{figure}[t] \centering
		\includegraphics[width = 0.9 \textwidth]{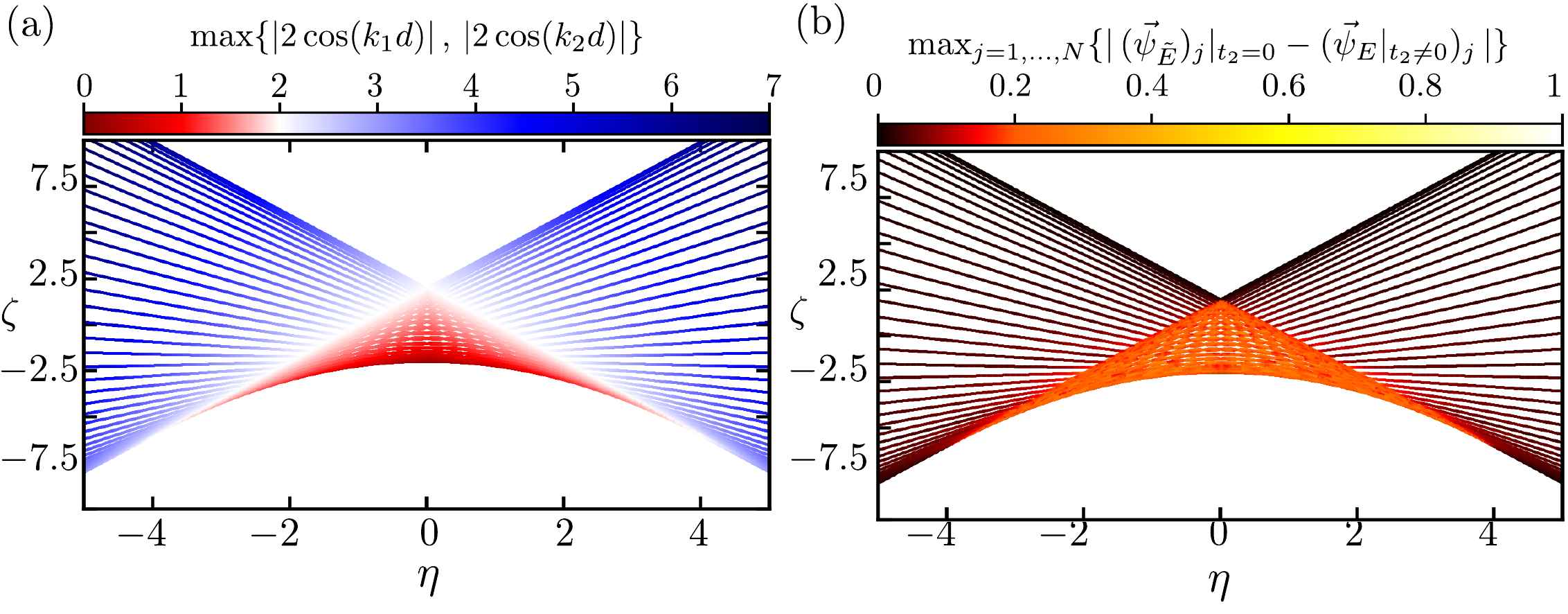}\vspace{-0.2cm}
		\caption{Numerical spectrum of \eref{equation: nnn chain matrix} for $N=40$. (a) Real values of both $k_{1,2}$ are limited to the Tetranacci arrow in red. Outside and shown in blue, one wavevector becomes complex. (b) Maximal deviation of eigenvector entries for $t_2$ and $t_2\neq 0$. The complex solution for $k$ enables the proper limit on the spectrum for $\eta\rightarrow \infty$ ($t_2\rightarrow 0$). Details are stated in the main text.}
		\label{Fig4}
	\end{figure}

	Naturally, we expect a continuous transition between $\vert t_2\vert \gg \vert  t_1\vert $ and $\vert t_2\vert \ll \vert t_1\vert $. At $t_2=0$ the spectrum of \eref{equation: nnn chain matrix} is $E=-\mu-2t_1\cos(kd)$ with $kd = n\pi/(N+1)$, $n = 1,\ldots,\,N$ \cite{Kouachi2008, Kouachi2006, Yueh2005, Gover1994, Shin1997} setting the tendency of the blue lines in Figure \ref{Fig4} (a). Yet, both $t_{1,2}$ influence $k_{1,2}$, particularly close to the arrow's boundary. Here, \eref{equation: equal energy constraint} suggests complex solutions of $k_{1,2}$ being apparent for $\vert t_1\vert> 4 \vert t_2 \vert$. In Figure \ref{Fig4} (a), the red (blue) color indicates two real (one complex and one real) wavevector. 
	
	 Analytically, the boundary conditions \eref{equation: nnc. bound1} - \eref{equation: nnc. bound4} explain the behavior. While $t_{1,2}\neq 0$ corresponds to four constraints and two left/ right mover contributions to $\xi_j$, the situation of $t_2 = 0$ implies $\xi_{0} = \xi_{N+1}=0$ and $\xi_j = A \exp(\rmi kd j)+ B \exp(-\rmi kd j)$. In order to satisfy the extreme case of $t_2 \rightarrow 0$, one wavevector becomes complex. With increasing imaginary part, the complex wavevectors contribution to $\xi_j$ diminishes as shown in Fig. \ref{Fig4} (b). For the numerical data, we diagonalized \eref{equation: nnn chain matrix} at $t_2\neq 0$, $t_2 = 0$, ordering eigenvectors from lowest to largest eigenvalues. The color indicates the largest absolute deviation of the corresponding normalized eigenvector entries. Of course, the comparison is fruitless within the Tetranacci arrow where $t_2$ dominates. A closer numerical investigation reveals that the difference diminishes on exponential scales (as expected from \eref{equation: tetranacci plane wave form}) and does not change abruptly to zero outside the arrow. Additionally, non of the eigenvectors of \eref{equation: nnn chain matrix} is an edgestate in agreement with the bulk edge correspondence \cite{Aguado-2017, Mong-2011}. 
	 
	 Despite the complex wavevector outside the Tetranacci arrow, all eigenvalues lie correctly within the frame set by the bulk dispersion relation and the first Brillouin zone. Recalling $S_{1,2}^2-\eta S_{1,2}-\zeta -2 = 0$ from Lemma \ref{lemma: exponential form for tetranaccis} and $S_{1,2} = 2\cos(k_{1,2}d)$ provides a direct link to Figure \ref{Fig2}. Due to the  $2\pi$ periodicity, the (real part of) $k_{1,2}$ lies between $-\pi/d$ and $\pi/d$. The two extremes provide the boundaries $\zeta = 2\mp 2\eta$ of the blue double fan from Figure \ref{Fig4} (a). In $\zeta-\eta$ space, both lines cross in the upper tip of the arrow at $(\eta,\,\zeta ) = (0,\,2)$. The two flanks follow $\zeta = 2\mp 2\eta$ when $\zeta \le 2$. For the lower boundary of the Tetranacci arrow, we notice that $S_{1,2}^2-\eta S_{1,2}-\zeta -2 = 0$ yields smaller values of $\zeta$ when $\vert S_{1,2}\vert <2$. The lowest parameterized curve follows from $\partial \zeta/\partial S_{1,2} = 0$ for fixed $\eta$. For $S_{1,2} = \eta/2$, we have $\zeta = -2-\eta^2/4$ applicable only for $\vert \eta \vert \le 4$ since $k_{1,2}$ are found real and reside within the first Brillouin zone. For both finite number of atoms $N$ and open boundary conditions, $\zeta \neq -2-\eta^2/4$ as otherwise $S_1 = S_2$ contradicts Bloch's theorem by \eref{equation: definition S12}.	

	\section{Engineering effective next nearest neighbor coupling}
	\label{section: Engineering effective next nearest neighbor coupling}
	Unfortunately the most interesting part of the model, i.e. the Tetranacci arrow, is hardly accessible in experiments. However, the basic ingredients of nearest neighbor processes and onsite degrees of freedom allow the engineering of effective (next) nearest neighbor bonds $\tilde{t}_1$ ($\tilde{t}_2$). Importantly, the ratio
	$\tilde{t}_1/\tilde{t}_2$ is easy to manipulate and tuning it yields arbitrarily large or small values. For the proof of principle, we consider the $X-Y$ chain in transverse magnetic field and the Kitaev chain.
	\begin{figure}[t] \centering
		\includegraphics[width = 0.6 \textwidth]{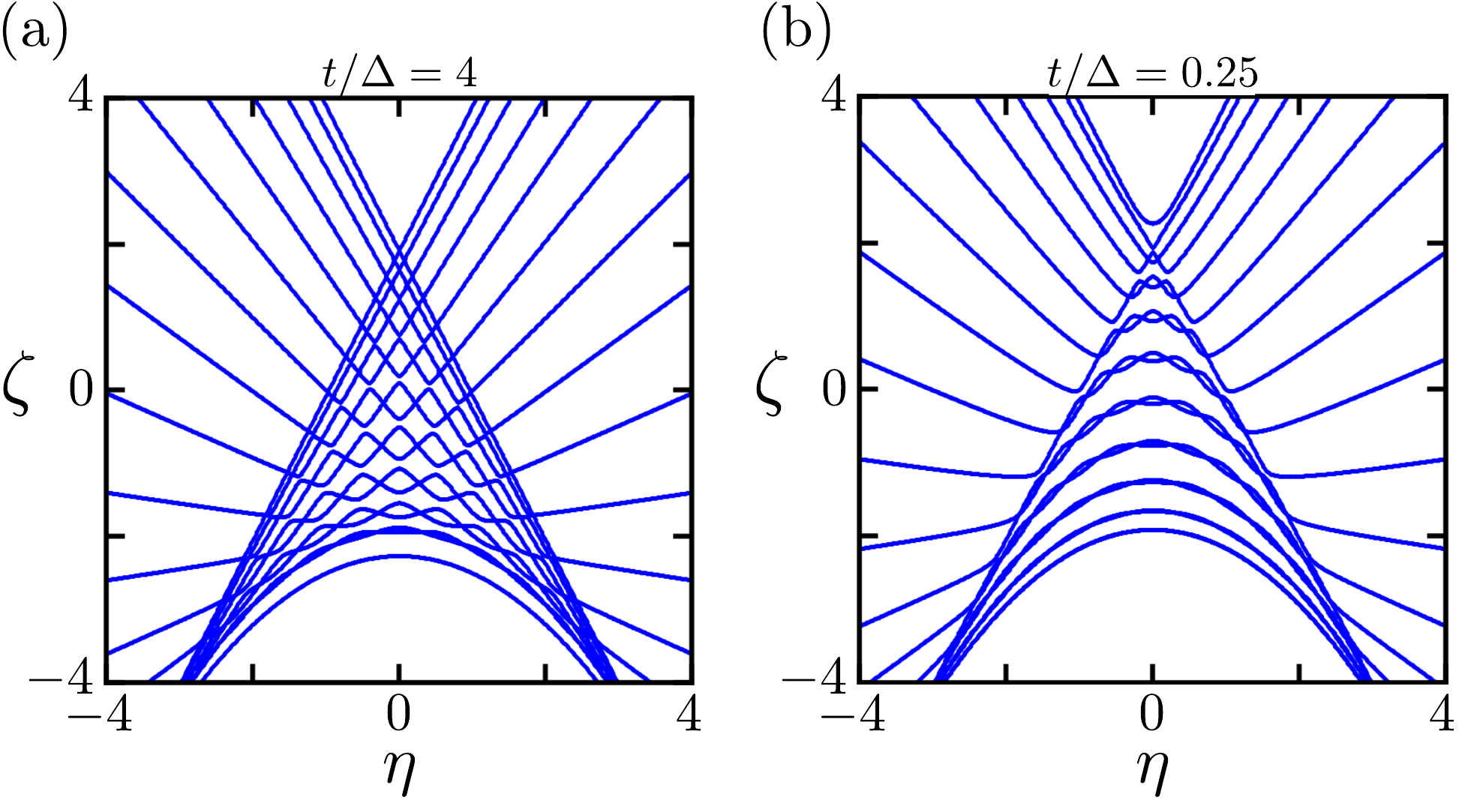}\vspace{-0.4cm}
		\caption{Numerical spectrum of the Kitaev chain in terms of $\zeta$, $\eta$ for $N=20$ and $t/\Delta = 4$ ($t/\Delta = 0.25$) in (a) ((b)). The spectrum differs for $\vert t/\Delta\vert>1$ and $\vert t/\Delta\vert<1$ associated to different limiting cases for $-\tilde{t}_1/\tilde{t}_2\rightarrow 0$. The Majorana fermions
		correspond to the lowest parabolic-like curve around $\eta = 0$ \cite{Leumer2020}.}
		\label{Fig5}
	\end{figure}

	The former consists of $N$ spins placed along a $1d$ axis and neighbors interact via their $x$ ($y$) components $S_j^\mathrm{x}$ ($ S_j^\mathrm{y}$) \cite{Lieb-1961}. Adding a transverse magnetic field $h$, the model reads \cite{Loginov1997}
	\begin{eqnarray}\label{equation: XY chain, real space, spin operators}
		\hat{H}_{\mathrm{XY}} = -h\sum\limits_{j=1}^N S_j^2 -\sum\limits_{j=1}^{N-1} \left(J_\mathrm{x} S_j^\mathrm{x} S_{j+1}^\mathrm{x}+ J_\mathrm{y} S_j^\mathrm{y} S_{j+1}^\mathrm{y}\right)
	\end{eqnarray}
	where $J_\mathrm{x}$, $J_\mathrm{y}$ abbreviate interaction constants. The diagonalization of \eref{equation: XY chain, real space, spin operators} was undertaken in \cite{Loginov1997} (\cite{Lieb-1961}) for $h\neq 0$ ($h=0$) using the Jordan-Wigner transformation which replaces spin operators by spin-less fermionic operators $c_j$, $c_j^\dagger$. The spin chain maps on the Kitaev chain as pointed out by Zvyagin \cite{Zvyagin-2015}. Therefore, we start directly from the latter. Spinless electrons experiencing hopping $t$ and $p$-wave superconductivity between nearest neighboring atoms, setting the Kitaev chain as \cite{Kitaev-2001}
	\begin{eqnarray}
		\label{equation: Kitaev chain, real space}
		\hat{H}_\mathrm{KC}&=-\mu \sum\limits_{j=1}^N\left(c_j^\dagger 	c_j-\frac{1}{2}\right)
		-t	\sum\limits_{j=1}^{N-1}\left(c_{j+1}^\dagger c_j+c_j^\dagger c_{j+1}\right)\nonumber\\&\quad  + \Delta\sum\limits_{j=1}^{N-1}\left(c_{j+1}^\dagger 	c_{j}^\dagger + c_{j+1} c_{j}\right).
	\end{eqnarray}
	We consider $t$, $\Delta$, $\mu \in \mathbb{R}$ and $N$ atoms. Independently of the $X-Y$-chain, this model has attracted some attention in the past as it is the archetypal model for topological superconductors. Although $p$-wave superconductivity is itself rare in nature, several promising platforms were proposed in the past decade in order to engineer the desired odd parity superconductivity \cite{Lutchyn-2010, Oreg-2010, Alicea-2012, Izumida-2017}. However, the hunt for Majorana fermions/ Majorana zero modes is not our motivation to subsequently introduce Majorana operators \cite{Kitaev-2001, Aguado-2017}
	\begin{eqnarray}
		\left(\matrix {\gamma^{\textcolor{blue}{A}}_j\cr
		\gamma^{\textcolor{orange}{B}}_j} \right) =\frac{1}{\sqrt{2}}\left(\matrix {1&1\cr
		-\mathrm{i} &\mathrm{i}} \right)	\left(\matrix{ c_j\cr
			c_j^\dagger} \right),
	\end{eqnarray}
	rather the aim for technical simplicity. The advantage is that Majorana operators  treat $t$ and $\Delta$ equally since they own particle and hole properties simultaneously: $(\gamma^{\textcolor{blue}{A}}_j)^\dagger = \gamma^{\textcolor{blue}{A}}_j$, $(\gamma^{\textcolor{orange}{B}}_j)^\dagger = \gamma^{\textcolor{orange}{B}}_j$. Then, \eref{equation: Kitaev chain, real space} becomes \cite{Kitaev-2001, Aguado-2017}
	\begin{eqnarray}\label{eq: Kitaev chain, Majorana operators}
		\hat{H}_\mathrm{KC} = -\mathrm{i}\mu \sum\limits_{j=1}^N\,\gamma^{\textcolor{blue}{A}}_j\,\gamma^{\textcolor{orange}{B}}_j + \mathrm{i}(\Delta -t) \sum\limits_{j=1}^{N-1}\,\gamma^{\textcolor{blue}{A}}_j\,\gamma^{\textcolor{orange}{B}}_{j+1}+ \mathrm{i}(\Delta +t) \sum\limits_{j=1}^{N-1}\,\gamma^{\textcolor{orange}{B}}_j\,\gamma^{\textcolor{blue}{A}}_{j+1}
	\end{eqnarray}
	Aiming on the spectrum of \eref{eq: Kitaev chain, Majorana operators} suggest the introduction of Majorana sublattices $\hat{\Psi}_{\alpha} = \left(\gamma^{\alpha}_{1},\,\ldots,\,\gamma^{\alpha}_{N}\right)^\mathrm{T}$, $\alpha = \textcolor{blue}{A},\,\textcolor{orange}{B}$. In this basis, the eigenvectors entries are revealed as symmetric Tetranacci polynomials with coefficients
	$\zeta = (E^2-\mu^2-2t^2-2\Delta^2)/(t^2-\Delta^2)$, $\eta = -2t\mu/(t^2-\Delta^2)$ \cite{Leumer2020, Leumerthesis} and $E$ as eigenvalue. Comparing $\eta = -\tilde{t}_1/\tilde{t}_2$ with the lattice model from \sref{section: lattice model}, we have effective (next) nearest neighbor hopping \emph{processes} $\tilde{t}_1 = 2t\mu$ ($\tilde{t}_2 = t^2-\Delta^2$), ignoring the wrong physical dimension of $\tilde{t}_{1,2}$. The ratio (value of) $\tilde{t}_1/\tilde{t}_2$ ($\eta$) could be tuned by an electrostatic gate shifting the potential; thus, modifying the onsite energy $\mu$. From \cite{Loginov1997}, we have $\tilde{t}_1 = -2h (J_x+J_y)$, $\tilde{t}_2 =J_{\rm x}J_{\rm y}$ for the $X-Y$-chain. Henceforth, a small magnetic field in $z$ direction reflects strong effective next nearest neighbor coupling. We finish this section discussing the emergence of the Tetranacci sequence and thus the engineering of $\tilde{t}_{1,2}$.
	\begin{figure}[t] \centering
		\includegraphics[width = 0.7 \textwidth]{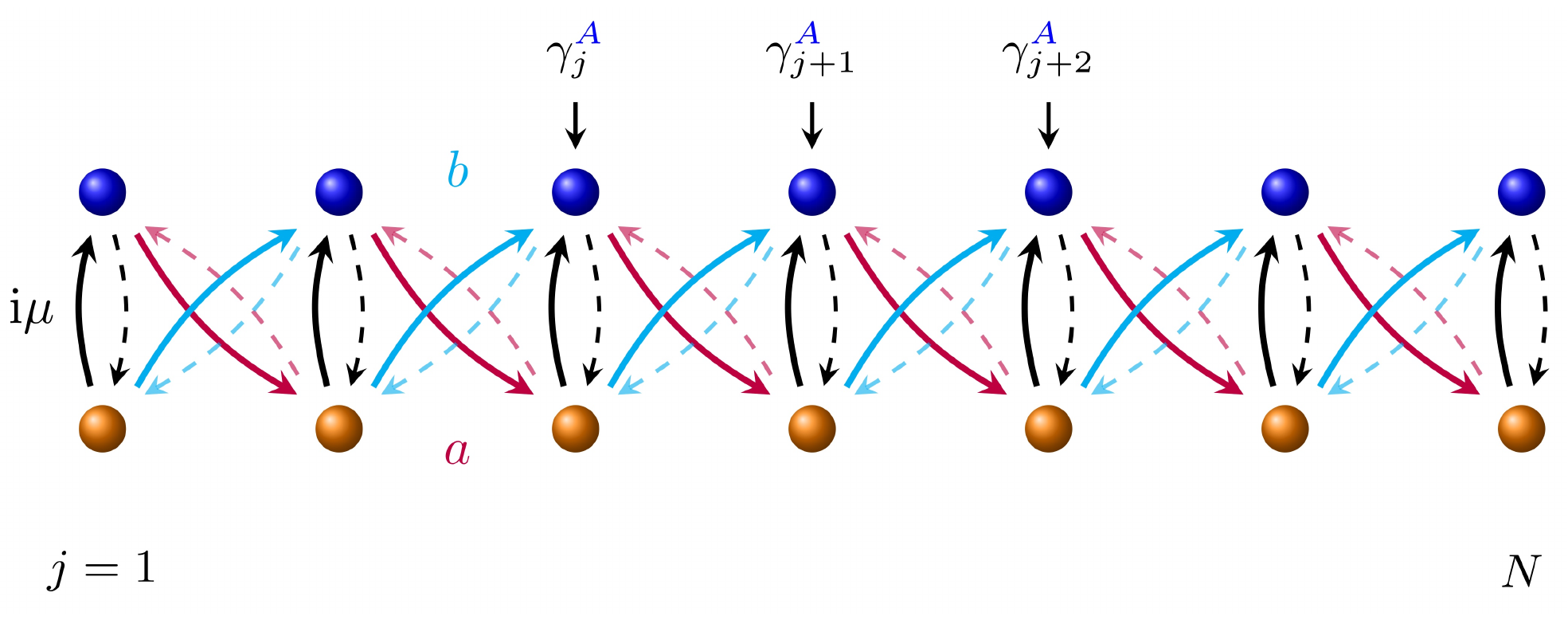}\vspace{-0.4cm}
		\caption{Kitaev chain in Majorana sublattice basis for $N=7$ atoms. Solid (dashed) arrows correspond to $\mathrm{i}\mu$ in black, $a = \mathrm{i}(\Delta -t)$ ($-a$) in purple and $b = \mathrm{i}(\Delta +t)$ (-b) in blue. Only the horizontal axis has a spatial extend with $d$ as the interatomic distance. The vertical axis displays onsite degrees of freedoms. }
		\label{Fig6}
	\end{figure}

	A first impression can be taken from Fig. \ref{Fig5}, where the numerically found spectrum forms the Tetranacci arrow. Notice though that $\tilde{t}_1/\tilde{t}_2\rightarrow \infty$ has distinct limits due the substructure of $\tilde{t}_2$. Analytically, one can avoid the cumbersome algebra from \cite{Leumer2020} and deduce the proper result pictorially from the system's sketch in Figure \ref{Fig6}. The diagonalization of \eref{eq: Kitaev chain, Majorana operators} implies a disentanglement of the sublattices; thus projecting the information of sublattice $\hat{\Psi}_{\textcolor{orange}{B}}$ on  $\hat{\Psi}_{\textcolor{blue}{A}}$ (or vice versa). This demands merely to identify all non repetitive processes from some $\gamma_j^{\textcolor{blue}{A}}$ to some $\gamma_{l}^{\textcolor{blue}{A}}$ ($j,l = 1,\ldots,N$) by following the arrows and multiplying the associated constants. Ignoring initially all boundary effects, we use $\gamma_j^{\textcolor{blue}{A}}$ as starting position as sketched in Figure \ref{Fig6}. Spatial inversion symmetry allows to focus only on forward processes, i.e. on growing $j$. For nearest neighbors $\gamma_j^{\textcolor{blue}{A}}\rightarrow\gamma_{j+1}^{\textcolor{blue}{A}}$, we have $\downarrow \textcolor{cyan}{\nearrow} \equiv -\mathrm{i}\mu\, b$ and $\textcolor{purple}{\searrow}\,\uparrow\equiv a\,\mathrm{i}\mu$. For $\gamma_j^{\textcolor{blue}{A}}\rightarrow\gamma_{j+2}^{\textcolor{blue}{A}}$, we have $\textcolor{purple}{\searrow}\textcolor{cyan}{\nearrow}  = ab$ only, without repetitive use of effective nearest neighbor bonds. The effective onsite terms $\gamma_j^{\textcolor{blue}{A}}\rightarrow\gamma_{j}^{\textcolor{blue}{A}}$ are $\downarrow \uparrow = \mu^2$, $\textcolor{purple}{\searrow} \textcolor{purple}{\nwarrow} = -a^2$ $\textcolor{cyan}{\swarrow} \textcolor{cyan}{\nearrow} = -b^2$; merely a corrective factor has to be added when the chain has finite length. Collecting all information and accounting for the backward processes yields ($j, j' = 1\,,\ldots,\,N$)
	\begin{eqnarray}\label{The Kitaev chain, eq: pictorial hhdagger}
		h_{jj'} &= \left[\mu^2-a^2(1-\delta_{jN})-b^2(1-\delta_{j1})\right]\,\delta_{j,j'}\,+\,\mathrm{i} \mu(a-b)\,\left(\delta_{j,j'+1}\,+\,\delta_{j+1,j'}\right)\nonumber\\
		&\quad +\,ab\,\left(\delta_{j,j'+2}\,+\,\delta_{j+2,j'}\right).
	\end{eqnarray}
	The spectrum of \eref{eq: Kitaev chain, Majorana operators} is found from $h\vec{v}_{\textcolor{blue}{A}} = E^2\vec{v}_{\textcolor{blue}{A}}$ with $(\vec{v}_{\textcolor{blue}{A}}, \vec{v}_{\textcolor{orange}{B}})$ as eigenvector of \eref{eq: Kitaev chain, Majorana operators}. The entries of both sublattice vectors $\vec{v}_{{\textcolor{blue}{A}}, {\textcolor{orange}{B}}}$ are symmetric Tetranacci polynomials with the earlier mentioned coefficients. Although $h$ has the same structure as $M(\alpha\neq 0, \beta\neq0 )$ from \eref{equation: Toeplitz like matrix} (and is thus not a Toeplitz matrix) $h\vec{v}_{\textcolor{blue}{A}} = E^2\vec{v}_{\textcolor{blue}{A}}$ can be treated along the lines of section \ref{section: lattice model} \cite{Leumer2020}.	
	\section{Conclusion}
	\label{section: conclusion}
	Subsequent to the definition of symmetric Tetranacci polynomials $\xi_j$ ($j \in \mathbb{Z}$), we gave a closed form expression in terms of basic Tetranacci polynomials $\mathcal{T}_i(j)$ ($i=-2,\ldots,\,1$). Due to their initial values $\mathcal{T}_i(l) = \delta_{il}$ ($l = -2,\,\ldots,1\,$), they inherit specific traits and interrelations. In turn, $\xi_j$ can be constructed from $\mathcal{T}_{-2}$ alone. We demonstrated in \eref{equation: closed form expression of T-2} the decomposition of $\mathcal{T}_{-2}(j)$, and thus generic $\xi_j$, in terms of the Fibonacci/ Tetranacci polynomials $\varphi_{1,2}(j)$. The sinusoidal representation of $\varphi_{1,2}(j)$ enables a standing wave form of $\xi_j$ in \eref{equation: final closed form for xi_j for S1 not S2}; thus reflecting already their potential in condensed matter physics when appropriate prerequisites are met. We discussed that Tetranacci polynomials generally allow for complex wavevectors $k_{1,2}$ independently of a physical model's topological classification. Contrary to previous works (cf. \cite{Leumer2020, Leumer2021}), we both generalized and simplified the presented results beyond $S_1\neq S_2$. Based on the rigorous mathematical approach, the shown results are generic and generalize beyond a concrete physical system; thus, extending earlier limitations. 	
	
	Besides the theoretical treatment of Tetranacci polynomials, we have shown their appearance in three physical systems: The $X-Y$ model in transverse magnetic field, the Kitaev chain and in atomic tight-binding systems owning nearest and next nearest neighbor hopping. For the latter, we explicitly demonstrated that Tetranacci polynomials $\xi_j$ set the wave function at an atomic position $j$. Solving the associated eigenvector equation, we derived a transcendental quantization condition for the wavevectors $k_{1,2}$, cf. \eref{equation: quantization contraint nnn chain}. Typically, the solution differs from the naively expected particle in the box behavior, that is $k = n\pi/(N+1)$ ($n=1,\ldots,\,N$) for $N$ atoms. The three main reasons are, firstly that $k_{1,2}$ depend on the models parameters by means of the equal energy constraint $E(\pm k_1) =E(\pm k_2)$. Secondly, the quantization constraint captures the next nearest neighbor character of the model, i.e. $N+2$ rather than $N+1$ appears. Lastly, $k_{1,2}$ are generally complex. The momentum quantization condition yields a degenerate spectrum at well defined parameter values. We deduced and stated the conditions. Spatial inversion symmetry protect the degenerate eigenvalues and we demonstrated that the spatial parity of the associated eigenstates is imprinted into the transcendental constraint  \eref{equation: quantization contraint nnn chain}. 
	
	This proofed fatal since symmetry relation fail to simplify the expression of eigenvectors beyond a merely formal solution due to unknown non-linear identities (if existent) of $\mathcal{T}_{-2}(j)$. Quantum transport (cf. \ref{appendix: quatum transport}) suffers the same fate although model specific traits may allow progress in the linear transport regime as has been demonstrated earlier for the Kitaev chain \cite{Leumer2021}.
	
	When the Tetranacci recursion is obeyed, the spectrum of quantum devices features a shape we dubbed as \emph{Tetranacci arrow}. Inside, wavevectors are real and all degenerated eigenvalues are placed within. Experimental access to the Tetranacci arrow is seemingly denied, as the next nearest neighbor hopping $t_2$ has to dominate its nearest neighbor cousin $t_1$. However, coupled onsite degrees of freedom and nearest neighbor processes allows to engineer effective (next) nearest neighbor bonding $\tilde{t}_1$ ($\tilde{t}_2$) with arbitrarily large or small ratio $\tilde{t}_1/\tilde{t}_2$. We provided a descriptive proof of principle based on the $X-Y$-model and the Kitaev chain. For the former, we argued that $\tilde{t}_1/\tilde{t}_2$ depends on the onsite energy $\mu$ and is thus  experimentally tune able by an electronic gate. For the $X-Y$ chain this ratio depends on the external magnetic field, where small field strength relate to large $\tilde{t}_1/\tilde{t}_2$.
	
	Several setups to simulate the Kitaev chain in table-top experiments using macroscopic elements were proposed and realized \cite{Liu2022, Allein2022}. Particularly, the usage of magnetic spinners to fabricate a macroscopic version of a quantum ladder provides an mechanical analog of the Kitaev chain with an adjustable onsite term \cite{Qian2023}. Besides, unpaired electrons in carbon ladder polymers exhibit ferromagnetic and antiferromagnetic interactions such that effective $\tilde{t}_{1,2}$ bonds are formed \cite{Ortiz2023}. The discussed physics of an atomic chain featuring nearest and next nearest neighbor hopping is therefore experimentally accessible and allows the on-demand facilitation of complex wavevectors. In short terms, the Tetranacci arrow is real indeed.
	\ack
	We gratefully acknowledge financial support from the Deutsche Forschungsgemeinschaft (SFB 1277 Project B04), the Elitenetzwerk Bayern (IGK Topological Insulators) and from the French National Research Agency ANR through project ANR-20-CE30-0028-01. We acknowledge interesting and fruitful discussion with M. Grifoni, M. Marganska, D. Weinmann and S. Hrdina. Further, we thank W. H{\"a}usler and M. Nieper-Wisskirchen for encouraging us to write this manuscript. Finally, we are grateful for all suggestions by the referees.
	\appendix
	\section{Degenerated roots $r_{\pm 1,2}$}
	\label{appendix: degenerated roots}
	\begin{prop} Degeneracies of $r_{\pm l}$ grant 
		\begin{eqnarray}
			r_{\pm 1}^j,\,j\,r_{\pm 1}^j,  &S_1 =S_2,\,S_1^2\neq 4,\\
			r_{+ 1}^j,\,j\,r_{+ 1}^j,\, j^2\,r_{+ 1}^j,\,j^3\,r_{+ 1}^j,\qquad&S_1 =S_2,\,S_1^2= 4
		\end{eqnarray}
		as additional solutions to \eref{equation: Tetranacci recursion formula}.
	\end{prop}

	\begin{proof}
		We focus first on $S_1 =S_2,\,S_1^2\neq 4$, where \eref{equation: definition rpml in the power ansatz} implies $r_{\pm 1} = r_{\pm 2}$. Substituting $\xi_j \propto j\,r_{\pm 1}^j$ into \eref{equation: Tetranacci recursion formula} and dividing by $r_{\pm 1}^{j-2}\neq 0$ gives
		\begin{eqnarray}\label{appendix equation: jrj is solution}
		\fl \qquad	0 = j\left[r_{\pm 1}^4 - \zeta r_{\pm 1}^2+1-\eta \left(r_{\pm 1}^3+r_{\pm 1}\right)\right]\,+\,2\,r_{\pm 1}^4-2 - \eta\,(r_{\pm 1}^3-r_{\pm 1}).
		\end{eqnarray}
		Here, the first bracket vanishes since $r_{\pm 1}$ satisfies \eref{equation: Tetranacci recursion formula}. Due to \eref{equation: definition S12}, $S_1 = S_2$ implies $\eta = 2S_1 $. Using $S_1 = r_{+1}+r_{-1}$, one has
		\begin{eqnarray}\label{appendix equation: jrj is solution2}
		\fl \qquad	0 = 2\,r_{\pm 1}^4-2 - \eta\,(r_{\pm 1}^3-r_{\pm 1})  = 2\,r_{\pm 1}^4-2 - 2 (r_{+1}+r_{-1})\,(r_{\pm 1}^3-r_{\pm 1}).
		\end{eqnarray}
		Recalling $r_{+1}r_{-1} = 1$ and \eref{appendix equation: jrj is solution} is satisfied indeed. Thus, $\xi_j\propto j\,r_{\pm 1}^j$ is a symmetric Tetranacci polynomial.
		
		For $S_1 =S_2,\,S_1^2= 4$, \eref{equation: definition rpml in the power ansatz}, states $r_{+1} = r_{-1} = r_{+2}=r_{-2}$. Due to $S_1 = S_2$, $jr_{+1}^j$ is already a solution and we turn directly to $j^2r_{1}^j$. From \eref{equation: Tetranacci recursion formula}, we find		
		\begin{eqnarray}
		\fl \qquad	0 &= j^2\left[r_{\pm 1}^4 - \zeta r_{\pm 1}^2+1-\eta \left(r_{\pm 1}^3+r_{\pm 1}\right)\right]\,+\,2j\left[2\,r_{\pm 1}^4-2 - \eta\,(r_{\pm 1}^3-r_{\pm 1})\right],\nonumber\\
		\fl \qquad	&\qquad +4\,r_{\pm 1}^4+4 - \eta\,(r_{\pm 1}^3+r_{\pm 1}).
		\end{eqnarray}
		Here, the first line vanishes. Since \eref{equation: definition rpml in the power ansatz} implies  $r_{+1} = S_1/2$ at $S_1 = S_2$, $S_1^2=4$, we find 
		\begin{eqnarray}
		\fl \qquad	4\,r_{\pm 1}^4+4 - \eta\,(r_{\pm 1}^3+r_{\pm 1})\,=\,4 \left(\frac{S_1}{2}\right)^4-4 \frac{S_1}{2}\left(\frac{S_1}{2}\right)^3+4-S_1^2 = 0.
		\end{eqnarray}
		Thus, $j^2\,r_{+1}^j$ is a solution to \eref{equation: Tetranacci recursion formula}. Further, $j^3\,r_{+1}^j$ also satisfies \eref{equation: Tetranacci recursion formula}. Similar as before, terms associated to $j^2, j,j^0$ drop. We arrive at
		\begin{eqnarray}
		\fl \qquad	0 \,=\,8r_{+1}^4-8-\eta\, r_{+1}^3+\eta\, r_{+1}
			&= 8 \left(\frac{S_1}{2}\right)^4 - 8 -\eta\,r_{+1}\left(r_{+1}^2-1\right)\nonumber\\
			&= 8 \left(\frac{S_1}{2}\right)^4 - 8 -\eta\,r_{+1}\left(\frac{S_1^2}{4}-1\right) \equiv 0,
		\end{eqnarray}
		due to $r_{+1} = S_1/2$, $S_1^2 = 4$. Hence, $j^3\,r_{+1}^{j}$ satisfies \eref{equation: Tetranacci recursion formula}. \hfill $\square$
	\end{proof}
	
	\section{$\mathcal{T}_{-1}(j)$, $\mathcal{T}_{0}(j)$, $\mathcal{T}_{1}(j)$ for degenerate roots}
	\label{appendix: Formulae of T-1,.., T_1 for degenerate roots}
	For $S_1 = S_2\neq \pm 2$, we have
	\begin{eqnarray}
	\fl \qquad	\mathcal{T}_{-1}(j) &= \frac{3j \varphi_1(j+2)- (j+2)(S_1^2-1)\varphi_{1}(j)}{S_1^2-4},\\
	\fl \qquad	\mathcal{T}_{0}(j) &= \frac{2(S_1^2-1)(j+2)\,\varphi_1(j+1)- 3(j+1)S_1\,\varphi_1(j+2)}{S_1^2-4},\\
	\fl \qquad	\mathcal{T}_{1}(j) &= \frac{j\varphi_1(j+2)-(j+2)\varphi_{1}(j)}{S_1^2-4},
	\end{eqnarray}
	while $S_1 = S_2 = \pm2$ gives
	\begin{eqnarray}
	\fl \qquad	\mathcal{T}_{-1}(j) &= S_1\frac{(2-j)j\,\varphi_1(j-1) + 2S_1 (j^2-1)\,\varphi_1(j)}{12},\\
	\fl \qquad	\mathcal{T}_{0}(j) &= S_1\frac{(3+j)(1+j)\,\varphi_1(j+2) -2S_1 (j+2)j\,\varphi_1(j+1)}{12},\\
	\fl \qquad	\mathcal{T}_{1}(j) &= S_1\frac{(2+j)j\, \varphi_1(j+1)}{12}
	\end{eqnarray}
	for generic integer $j$.
	\begin{proof}
		The displayed formulae follow directly by substituting \eref{equation: closed form expression of T-2} into the relations from Lemmata \ref{lemma: inversion relation of Tij}, \ref{lemma: further interconnection of the basic Tetranacci polynomials} and exploiting the properties of $\varphi_{1,2}(j)$ drawn in Proposition \ref{proposition: varphi is odd in j} and Theorem \ref{theorem: hidden Fibonacci polynomials}. \hfill $\square$
	\end{proof}
	\section{Tetranacci polynomials in quantum transport}
	\label{appendix: quatum transport}
	When applicable, Tetranacci polynomials own the potential to provide exact analytic results for the linear and non-linear transport of quantum and mesoscopic devices. However, the unknown non-linear identities between Tetranacci polynomials (if existent) impose serious challenges in simplifying the retarded Green's function (see below). Captured by the explicit structure of the coefficients $\zeta$, $\eta$, model specific properties may still allow substantial progress as has been shown in case of the Kitaev chain \cite{Leumer2021, Leumerthesis}. Unfortunately, this seems not to be true for the atomic chain discussed in \sref{section: lattice model}. However, the system actually allows for a short and evident connection between Tetranacci polynomials, the retarded Green's function $G^r$ and the transmission probability $\mathcal{T}(E)$. We provide a \underline{sketch} on the basic strategy.
	
	We consider two normal conducting, non interacting contacts ($\alpha = L, R$)
	\begin{eqnarray}
		\hat{H}_\alpha = \sum\limits_k \epsilon_{k\alpha}\,c_{k\alpha}^\dagger c_{k\alpha}
	\end{eqnarray}
	sandwiching the atomic chain $\hat{H}$ stated in \eref{equation: nnn chain, lattice Hamiltonian}. Here, $\epsilon_{k\alpha}^{(\dagger)}$ removes (adds) a spinless electron from (to) contact $\alpha$ and $\epsilon_{k\alpha}$ denotes the electron's energy. The tunneling Hamiltonian
	\begin{eqnarray}
		\hat{H}_{\rm T} = \sum\limits_k \left(t_L(k) \,d_1^\dagger c_{kL}\, +\, t_R(k)\, d_N^\dagger c_{kR}\right) \,+\,{\rm h.c.}
	\end{eqnarray}
	allows the exchange of particles only close to the chain's ends for simplicity. An applied bias $eV$ between the two contacts initializes an electric current $I(t) = -e \,d\langle \hat{N}_\alpha \rangle/(dt)$, $e>0$. The Hamiltonian $\hat{H}_{\rm tot} = \hat{H}+ \hat{H}_{\rm T} + \hat{H}_{L}+\hat{H}_{R}$ captures the time evolution of the particle number operator $\hat{N}_\alpha =\sum_k  c_{k\alpha}^\dagger c_{k\alpha}$. Applying the non-equilibrium Green's function technique yields straightforwardly the steady state current  \cite{Flensberg2004, Meir-1992, Ryndyk-2016, Haug-1996}
	\begin{eqnarray}\label{current equation}
		I = \frac{e}{h}\int\limits_{-\infty}^{\infty} dE~ \mathcal{T}(E)\left[f(E)-f(E+eV)\right],
	\end{eqnarray}
	after some algebra.	Here, $f(E)$ abbreviates the Fermi function and $h$ is Planck's constant. The self-energies $\Sigma^r_{L,R}$ ($i,j = 1,\,\ldots,\,N$)
	\numparts
	\begin{eqnarray}
		(\Sigma^r_L)_{i,j} &= (\Lambda_L-\mathrm{i}\gamma_L) \,\delta_{1i}\delta_{1j},\\
		(\Sigma^r_R)_{i,j} &= (\Lambda_R-\mathrm{i}\gamma_R) \,\delta_{Ni}\delta_{Nj}
	\end{eqnarray}
	\endnumparts
	account for the coupling of leads and chain. They are sparse matrices due to the choice of $\hat{H}_{\rm T}$. The quantities $\gamma_\alpha = \pi \sum_k \vert t_\alpha(k)\vert^2 \,\delta(E-\epsilon_{k\alpha})$, $\Lambda_\alpha = \mathcal{P} \sum_k \vert t_\alpha(k)\vert^2/(E-\epsilon_{k\alpha})$ are real and $\mathcal{P}$ denotes Cauchy's principle value. The transmission probability $\mathcal{T}(E) = {\rm Tr}\{\Gamma_L G^r\Gamma_R G^a\}$ is given as a trace over broadening matrices $\Gamma_\alpha = -2\, \mathrm{Im}(\Sigma^r_\alpha)$ and retarded (advanced) Green's functions $G^r$ ($G^a$). Due to $G^a = (G^r)^\dagger$, $G^r = (E\mathds{1}_N- \mathcal{H}-\Sigma_L^r-\Sigma_R^r)^{-1}$ and $\mathcal{H}$ from \eref{equation: nnn chain matrix}, $\mathcal{T}(E)$ is fully determined.
	
	Contrary to the full non-linear transport regime in \eref{current equation}, the linear conductance
	\begin{eqnarray}
		G = \lim\limits_{eV\rightarrow 0} \frac{\partial I}{\partial V} = \frac{e^2}{h} \mathcal{T}(E = 0).
	\end{eqnarray}
	relates to $\mathcal{T}(0)$ only for zero temperature. The sparsity of $\Gamma_{L,R}$ yields generally $\mathcal{T}(E) = 4 \gamma_L\gamma_R \vert G^r_{1N}\vert^2$. Instead of performing the inversion, $G^r_{1N}$ can be obtained much easier using Tetranacci poylnomials. Since the connection is more apparent for $E\neq 0$, we focus on $\mathcal{T}(E)$ instead. 
	
	Obviously, $G^r$ obeys $(E\mathds{1}_N- \mathcal{H}-\Sigma_L^r-\Sigma_R^r)G^r = \mathds{1}_N$. As we shall see, we can approach as for the eigenvector equation of $ \mathcal{H}$. However, notice that $E$ is a continuous variable and not necessarily an eigenvalue of $\mathcal{H}$. When $G^r = (\vec{v}_1,\,\ldots,\,\vec{v}_N)^\mathrm{T}$ is written in terms of column vectors $\vec{v}_i \in \mathbb{R}^{N\times 1}$, those obey ($i = 1\,\ldots,\,N$)
	\begin{eqnarray}\label{equation: column vectors of G^r}
		(E\mathds{1}_N- \mathcal{H}-\Sigma_L^r-\Sigma_R^r) \vec{v}_i = \vec{e}_i
	\end{eqnarray}
	where $\vec{e}_i$ abbreviates the $i$-th column of $\mathds{1}_N$. The structure of \eref{equation: column vectors of G^r} reminds an eigenvector equation for the matrix $\mathcal{H}+\Sigma_L^r+\Sigma_R^r$, but the inhomoginity $\vec{e}_i$ yields a unique solution for $\vec{v}_i$. Aiming on $G^r_{1N}$, sets $i= N$ in \eref{equation: column vectors of G^r}. Defining $\vec{v}_N = (\sigma_1,\ldots,\,\sigma_N)^\mathrm{T}$ grants $\sigma_{j+2} =  \zeta\, \xi_j-\xi_{j-2} + \eta\, (\xi_{j+1}+\xi_{j-1})$ with $\zeta = -(E+\mu)/t_2$, $\eta = -t_1/t_2$ for $j = 3\,\ldots,N-3$ similar to the eigenvector equation of $\mathcal{H}$. Yet, the self-energies modify the boundary conditions. Extending again the recursion to all integers $j$, but keeping $\vec{v}_N$ untouched, the constraints reduce to
	\numparts
	\begin{eqnarray}
		\label{equation: GF boundary condition 1}
		0=&\sigma_0\\
		\label{equation: GF boundary condition 2}
		0=&\sigma_{N+1},\\
		\label{equation: GF boundary condition 3}
		0=&(\mathrm{i}\gamma_L -\Lambda_L)\sigma_1 -t_2\sigma_{-1},\\
		\label{equation: GF boundary condition 4}
		1=&(\mathrm{i}\gamma_R -\Lambda_R)\sigma_N -t_2\sigma_{N+2},
	\end{eqnarray}
	\endnumparts
	for $t_2\neq 0$. Here "$1$" originates from $(\vec{v}_N)_N = 1$. Due to Corollary \ref{corollary: closed form for xij}, we search for $g_{-2}\,\ldots,\,g_{-1}$. From \eref{equation: GF boundary condition 1}-\eref{equation: GF boundary condition 4}, we find
	\begin{eqnarray}\label{solution for sigma j}
		\sigma_j = g_{-2}\mathcal{T}_{-2}(j) + g_1 \left[\mathcal{T}_1(j) + \frac{\mathrm{i}\gamma_L - \Lambda_L}{t_2}\mathcal{T}_{-1}(j)\right]
	\end{eqnarray}
	thus a $2\times 2$ matrix equation
	\begin{eqnarray}\label{boundary 2x2 system}
		\left(\matrix{%
			a & b\cr	
			c & d
		}ß\right) \left(\matrix{%
	g_{-2}\cr	
		g_1
	}ß\right) = \left(\matrix{%
0\cr	
1
}ß\right)
	\end{eqnarray}
	with coefficients
	\numparts
	\begin{eqnarray}
		a &= \mathcal{T}_{-2}(N+1),\\
		b &= \mathcal{T}_{1}(N+1) + \frac{\mathrm{i}\gamma_L- \Lambda_L}{t_2}\mathcal{T}_{-1}(N+1),\\
		c &= (\mathrm{i}\gamma_R-\Lambda_R)\mathcal{T}_{-2}(N) - t_2 \mathcal{T}_{1}(N+2),\\
		d & = (\mathrm{i}\gamma_R-\Lambda_R) \left[\mathcal{T}_{1}(N) + \frac{\mathrm{i}\gamma_L- \Lambda_L}{t_2}\, \mathcal{T}_{-1}(N)\right]\nonumber\\
		&\quad -t_2\left[\mathcal{T}_{1}(N+2) + \frac{\mathrm{i}\gamma_L- \Lambda_L}{t_2}\, \mathcal{T}_{-1}(N+2)\right].
	\end{eqnarray}
	\endnumparts
	The solution to \eref{boundary 2x2 system} is
	\begin{eqnarray}\label{coefficients solution}
		\left(\matrix{%
			g_{-2}\cr	
			g_1
		}ß\right) = \frac{1}{ab-cd} \left(\matrix{%
		 -\mathcal{T}_{1}(N+1) - \frac{\mathrm{i}\gamma_L- \Lambda_L}{t_2}\mathcal{T}_{-1}(N+1)\cr	
		\mathcal{T}_{-2}(N+1)
	}ß\right).
	\end{eqnarray}
	Hence, $G^r_{1N} \equiv \sigma_N$ contains non-linear combinations of $\mathcal{T}_i(N),\,\mathcal{T}_i(N+1)$ ($i = -2,\,\ldots, 1$) when \eref{coefficients solution} is inserted into \eref{solution for sigma j}. 
	
	For zero temperature, the linear conductance depends on $\left.G^r_{1N}\right\vert_{E=0}$ and non-linear identities can possibly be circumvented. In case of the Kitaev chain, the setting $E = 0$ turns all Tetranacci into Fibonacci polynomials \cite{Leumer2021}. However, this is a model specific property imprinted into the respective composition of $\zeta$, $\eta$ in terms of the model's parameters. Unfortunately, the nearest neighbor hopping chain misses this feature and further progress towards an useful explicit formula for $\left.G^r_{1N}\right\vert_{E=0}$, $\mathcal{T}(0)$ is prohibited.

\end{document}